%% file: main.tex
\PassOptionsToPackage{x11names}{xcolor}
\documentclass[acmsmall,screen,nonacm]{acmart}

\PassOptionsToPackage{noend}{algorithmic}

\input{headers/math_commands.tex}

\input{headers/header-lqa}

\input{headers/header}
\input{headers/header-lqa-tail}

\newcommand{\todoterm}{\ensuremath{\mathtt{missing}}}

\setcopyright{acmlicensed}
\copyrightyear{2026}
\acmYear{2026}

\acmJournal{PACMPL}
\acmVolume{1}
\acmNumber{POPL}
\acmArticle{1}
\acmMonth{1}

\begin{document}

\title{Generative Compilation: On-the-Fly Compiler Feedback as AI Generates Code}

\author{Niels M\"undler-Sasahara}
\authornote{Both authors contributed equally to this research.}
\affiliation{%
  \institution{ETH Zurich}
  \city{Zurich}
  \country{Switzerland}
}
\email{niels.muendler@inf.ethz.ch}

\author{Hristo Venev}
\authornotemark[1]
\affiliation{%
  \institution{INSAIT and Sofia University ``St. Kliment Ohridski''}
  \city{Sofia}
  \country{Bulgaria}
}
\email{hristo.venev@insait.ai}

\author{Dawn Song}
\affiliation{%
  \institution{University of California, Berkeley}
  \city{Berkeley}
  \state{CA}
  \country{USA}
}
\email{dawnsong@berkeley.edu}

\author{Martin Vechev}
\affiliation{%
  \institution{ETH Zurich}
  \city{Zurich}
  \country{Switzerland}
}
\email{martin.vechev@inf.ethz.ch}

\author{Jingxuan He}
\affiliation{%
  \institution{University of California, Berkeley}
  \city{Berkeley}
  \state{CA}
  \country{USA}
}
\email{jingxuan.he@berkeley.edu}

\input{sections-new/abstract.tex}




\maketitle

\input{sections-new/intro.tex}

\input{sections-new/motivation.tex}
\input{sections-new/gencomp.tex}
\input{sections-new/lang.tex}
\input{sections-new/method.tex}
\input{sections-new/rust.tex}
\input{sections-new/eval.tex}
\input{sections-new/discussion.tex}
\input{sections-new/related.tex}
\input{sections-new/conclusion.tex}
\input{sections-new/acks.tex}


\bibliographystyle{ACM-Reference-Format}
\bibliography{references}

\input{sections-new/appendix.tex}


\end{document}

%% file: headers/math_commands.tex
\usepackage{amsmath,amsfonts,bm}

\renewcommand{\epsilon}{\varepsilon}
\renewcommand{\emptyset}{\varnothing}

\def\1{\bm{1}}

\def\eps{{\epsilon}}

\DeclareMathAlphabet{\mathsfit}{\encodingdefault}{\sfdefault}{m}{sl}
\SetMathAlphabet{\mathsfit}{bold}{\encodingdefault}{\sfdefault}{bx}{n}



%% file: headers/header-lqa.tex
\usepackage[utf8]{inputenc} 
\usepackage{lipsum} 

\usepackage[T1]{fontenc}

\usepackage{etoolbox}
\newbool{includeappendix}
\setbool{includeappendix}{true}

\input{headers/lqa/overfull}



\input{headers/lqa/subfigures}

\input{headers/lqa/colors}

\input{headers/lqa/listing}

%% file: headers/lqa/overfull.tex
\ifdefined\isoverfull
\else
\fi

%% file: headers/lqa/subfigures.tex
\usepackage{subcaption}


%% file: headers/lqa/colors.tex
\usepackage{xcolor}

\definecolor{my-full-blue}{HTML}{1F77B4}

\definecolor{my-full-orange}{HTML}{FF7F0E}
\definecolor{my-extra-orange}{HTML}{7C4514}

\definecolor{my-full-green}{HTML}{2CA02C}

\definecolor{my-full-red}{HTML}{d62728}

\definecolor{my-full-purple}{HTML}{9467bd}

\colorlet{my-blue}{my-full-blue!30}
\colorlet{my-orange}{my-full-orange!30}
\colorlet{my-green}{my-full-green!90}
\colorlet{my-red}{my-full-red!90}
\colorlet{my-purple}{my-full-purple!30}

\colorlet{mygreen}{PaleGreen3}
\colorlet{myred}{LightPink2}
\colorlet{myblue}{SteelBlue2}
\colorlet{mylightblue}{SteelBlue1}
\colorlet{myorange}{DarkOrange1}

%% file: headers/lqa/listing.tex
\usepackage{listings}
\usepackage{listings-rust}

\usepackage{textcomp}

\usepackage{xcolor}

\usepackage[scaled=0.8]{beramono}

\definecolor{ckeyword}{HTML}{7F0055}
\definecolor{ccomment}{HTML}{3F7F5F}
\definecolor{cstring}{HTML}{2A0099}

\lstdefinestyle{numbers}{
	numbers=left,
	framexleftmargin=20pt,
	numberstyle=\tiny,
	firstnumber=auto,
	numbersep=1em,
	xleftmargin=2em
}

\lstdefinestyle{layout}{
	frame=none,
	captionpos=b,
}

\lstdefinestyle{comment-style}{
	morecomment=[l]//,
	morecomment=[s]{/*}{*/},
	commentstyle={\color{ccomment}\itshape},
}

\lstdefinestyle{string-style}{
	morestring=[b]",%
	morestring=[b]',%
	stringstyle={\color{cstring}},
	showstringspaces=false,%
}

\lstdefinestyle{keyword-style}{
	keywordstyle={\ttfamily\bfseries},
	morekeywords={
		function,
		constructor,
		int,
		bool,
		return,
		returns,
		uint
	},
	morekeywords = [2]{},
	keywordstyle = [2]{\text},
	sensitive=true,
}

\lstdefinestyle{input-encoding}{
	inputencoding=utf8,
	extendedchars=true,
	literate=
	{ℝ}{$\reals$}1%
	{→}{$\rightarrow$}1%
	{α}{$\alpha$}1%
	{β}{$\beta$}1%
	{λ}{$\lambda$}1%
	{θ}{$\theta$}1%
	{ϕ}{$\phi$}1%
}

\lstdefinestyle{escaping}{
	moredelim={**[is][\color{blue}]{\%}{\%}},
	escapechar=|,
	mathescape=true
}

\lstdefinestyle{default-style}{
	basicstyle=\fontencoding{T1}\ttfamily\footnotesize,
	style=numbers,
	style=layout,
	style=comment-style,
	style=string-style,
	style=keyword-style,
	style=input-encoding,
	style=escaping,
	tabsize=2,
	upquote=true
}

%% file: headers/header.tex
\usepackage[english]{babel}
\hypersetup{
  colorlinks=true,
  linkcolor=myblue,
  citecolor=myblue,
  urlcolor=myblue
}
\usepackage{algorithm}
\usepackage{algorithmic}
\usepackage{color,soul}
\usepackage{ bbold }
\usepackage{multirow}
\usepackage{tabularx,booktabs,xltabular}
\usepackage{graphicx}
\usepackage{tabu}
\usepackage{array}
\usepackage{siunitx}
\usepackage{subcaption}
\usepackage{fontawesome}
\usepackage{stackengine}
\usepackage{svg}
\usepackage{mathtools}
\usepackage{xspace}
\usepackage{wrapfig}
\usepackage{makecell}
\usepackage{placeins}
\usepackage{float}
\usepackage{pifont}
\usepackage{svg}
\usepackage[most]{tcolorbox}
\usepackage{array}
\usepackage{dsfont}
\usepackage{enumitem}

\usepackage[nounderscore]{syntax}
\usepackage{semantic}

\newcolumntype{x}[2]{S[table-format=#1.#2,table-auto-round]}

\input{headers/lqa/tikz}
\usepackage{pgfplots}
\usepackage[capitalize]{cleveref}

\usepackage{pifont}
\newcommand{\lang}{\ensuremath{L}}

\newcommand{\Rustc}{the Rust compiler} 
\newcommand{\rustc}{\textsf{rustc}}

\newcommand{\FR}{FR\xspace}

\newcommand{\rimms}{\textsc{GC}}
\newcommand{\rimm}{\textsc{Generative Compilation}}
\newcommand{\gc}{\textsc{GC}}
\newcommand{\rimmfn}{\textsc{GC}\ensuremath{^{\mathsf{fn}}}}
\newcommand{\rfinal}{\textsc{Post Compilation}}
\newcommand{\rfinals}{\textsc{PC}}
\newcommand{\pc}{\textsc{PC}}
\newcommand{\vanillas}{\textsc{LLM}}
\newcommand{\vanilla}{\textsc{LLM}}

\newcommand{\dupgrade}{\textsc{UpdatedAPI}\xspace}
\newcommand{\dtranslation}{\textsc{Translation}\xspace}

\newcommand{\mgpt}{{GPT 5.3}\xspace}
\newcommand{\mgptcodex}{{GPT 5.3 Codex}\xspace}
\newcommand{\mopus}{{Claude Opus 4.8}\xspace}
\newcommand{\mopusshort}{{Opus 4.8}\xspace}
\newcommand{\mgemini}{{Gemini 3.5 Flash}\xspace}
\newcommand{\mgeminishort}{{Gemini 3.5}\xspace}
\newcommand{\mkimi}{{Kimi K2.7}\xspace}
\newcommand{\mglm}{{GLM 5.2}\xspace}
\newcommand{\mqwen}{{Qwen 3.5}\xspace}
\newcommand{\mqwenlarge}{{Qwen 397B}\xspace}
\newcommand{\mqwensmall}{{Qwen 9B}\xspace}

\definecolor{pastelgreen}{rgb}{0.47, 0.87, 0.47}
\definecolor{pastelred}{rgb}{1.0, 0.6, 0.6}

\renewcommand{\algorithmicrequire}{\textbf{Input:}}
\renewcommand{\algorithmicensure}{\textbf{Output:}}
\newcommand{\ttt}[1]{\text{\texttt{#1}}}

\usepackage{scalerel}

%% file: headers/lqa/tikz.tex
\usepackage{tikz}

\usetikzlibrary{arrows}
\usetikzlibrary{automata}
\usetikzlibrary{calc}
\usetikzlibrary{backgrounds}
\usetikzlibrary{decorations.markings}
\usetikzlibrary{decorations.pathmorphing}
\usetikzlibrary{decorations.pathreplacing}
\usetikzlibrary{fit}
\usetikzlibrary{patterns}
\usetikzlibrary{positioning}
\usetikzlibrary{shadows}
\usetikzlibrary{shapes}
\usetikzlibrary{shapes.geometric}
\usetikzlibrary{arrows.meta}

\definecolor{blue}{HTML}{347bc6}
\definecolor{green-underline}{HTML}{2de12c}
\definecolor{yellow-underline}{HTML}{ffd700}

\tikzstyle{block} = [
    minimum width=1cm,
    minimum height=0.5cm,
    align=center,
    fill=gray!30,
    rounded corners=5pt,
]

\tikzstyle{arrow} = [
	-{Latex[length=1.4mm, width=1.4mm]},
	draw=blue,
]

\tikzstyle{dashedline} = [
    draw=blue,
    dashed
]

\tikzstyle{line} = [
    draw=blue
]

\definecolor{lightblue}{RGB}{173,216,230} 

%% file: headers/header-lqa-tail.tex
\input{headers/lqa/references}

%% file: headers/lqa/references.tex
\usepackage[capitalize]{cleveref}

\crefformat{section}{\S#2#1#3}

\crefrangeformat{section}{\S#3#1#4\crefrangeconjunction\S#5#2#6}

\crefmultiformat{section}{\S#2#1#3}{\crefpairconjunction\S#2#1#3}{\crefmiddleconjunction\S#2#1#3}{\creflastconjunction\S#2#1#3}

\newcommand{\crefrangeconjunction}{--}

\crefname{listing}{Lst.}{listings}
\crefname{line}{Line}{Lines}
\crefname{appendix}{App.}{App.}

\newcommand{\app}[1]{%
	\ifbool{includeappendix}{\cref{#1}}{the appendix}%
}
\newcommand{\App}[1]{%
	\ifbool{includeappendix}{\cref{#1}}{The appendix}%
}

%% file: sections-new/abstract.tex
\begin{abstract}
Languages with rich static semantics, such as Rust, provide stronger guarantees for AI-generated code, but their strictness makes generation more difficult.
Off-the-shelf compilers can provide useful feedback post-generation, but does not guide intermediate generation steps, such as those during autoregressive LLM decoding.
Constrained decoding intervenes earlier by rejecting invalid tokens during sampling, but requires white-box model access and costly reimplementation for semantic constraints.

We introduce generative compilation, the first approach to obtaining compiler feedback on partial programs during generation.
The core technical device is a sealor: a lightweight, mostly syntax-guided transformation that converts partial programs into complete ones that standard compilers can diagnose.
It is designed such that possible-to-complete partial programs are never rejected, while preserving enough code context to catch genuine dead ends early.
We construct such a sealor on a core Rust-like calculus and prove that it satisfies these properties, all mechanized in Lean.
We extend it to the first partial-program checker for real Rust.

We evaluate our method on challenging repository-level Rust coding tasks, across both frontier black-box and open-weight models.
We show that generative compilation reduces non-compiling outputs and improves functional correctness, relative to standard post-generation feedback.
It does so by detecting a broad range of errors close to their source and early during generation, thereby reducing errors cascades and enabling focused diagnostics.
More broadly, generative compilation is a step toward making compilers a first-class citizen of AI-assisted programming active during generation, rather than a separate post-generation check.
\begingroup
\renewcommand{\thefootnote}{\textdagger}
\footnotetext{We publicly release our proof mechanization, code implementation, benchmarks, and evaluation results: \url{https://github.com/eth-sri/generative-compilation}.}
\endgroup
\end{abstract}

%% file: sections-new/intro.tex
\section{Introduction}
\label{sec:introduction}

AI has become a standard tool for code generation, driven by the growing capabilities of large language models (LLMs) and LLM-based coding agents \citep{jiang2024surveylargelanguagemodels}.
Yet AI-generated code remains error-prone, threatening the correctness and security of downstream software systems \citep{DBLP:conf/sp/PearceA0DK22,DBLP:conf/icml/VeroMCRB0HV25}.
This makes languages with rich static semantics particularly valuable for AI-generated code.
Rust, for example, uses its ownership and borrowing discipline to provide memory safety guarantees and prevent related reliability and security failures \citep{matsakis2014rust,rustbelt}.
At the same time, these guarantees are enforced through strict compile-time rules, making valid code harder to generate than in more permissive languages such as C \citep{rustswebench,DBLP:conf/icse/DeligiannisLMPR25,khatry2025crustbench}.
Compiler feedback can help bridge this gap by identifying violations and guiding the model toward code that satisfies these rules.

\paragraph{Post-Generation Compiler Feedback}
The most common way to leverage compiler feedback is to wait until entire files are fully generated, run the compiler as is on the resulting program, and feed any diagnostics back to the model.
This post-hoc feedback loop has important practical advantages: it works with black-box LLMs and benefits from the compiler's guarantees and diagnostics with no additional implementation cost~\citep{wang-etal-2022-compilable,bi-etal-2024-iterative,dou-etal-2024-stepcoder}.
However, it does not take advantage of the intermediate outputs produced during generation.
Feedback becomes available only after the complete program has been produced, meaning that every token emitted after the first unrecoverable error may be based on invalid assumptions and ultimately wasted.
Moreover, since the feedback concerns the entire output, the model receives a batch of accumulated diagnostics all at once, which can be large enough to hinder model understanding~\citep{du-etal-2025-context,liu-etal-2024-lost}.

\paragraph{Constrained Decoding}
Constrained decoding instead intervenes directly during generation.
It checks intermediate outputs of LLMs, accepting only tokens that keep the current output extendable to a valid program~\citep{synchromesh,DBLP:journals/tmlr/UgareSKM025,park2024gad,lipkin2025awrs,ts}.
Thus, if generation terminates, the final output is guaranteed to be valid.
However, for languages with expressive static semantics, achieving this guarantee requires substantial language-specific reimplementation, making it effectively impractical to support more than restricted fragments~\citep{ts,melcer2024constraineddecodingfillinthemiddlecode,chopchop}.
Moreover, token-level filtering is silent: the model is steered away from invalid continuations, but never receives compiler-style explanations that could help it revise the code.
Finally, constrained decoding requires white-box access to the model's next-token distribution and sampling process, ruling out its application to frontier models exposed only through black-box APIs \citep{DBLP:journals/tmlr/UgareSKM025,muendler2025constraineddiffusion}.

\paragraph{This Work: Generative Compilation}
We introduce \emph{generative compilation}, the first approach to make compiler-style feedback available \emph{during} code generation.
It serves as a middle ground between post-generation feedback and constrained decoding, addressing key limitations of both.
Like constrained decoding, a generative compiler processes partial programs while they are being produced.
Unlike constrained decoding, however, it provides compiler-style diagnostics upon rejection.
The model receives an explanation of what went wrong and can revise the generated code, rather than being silently steered away from invalid tokens.
Because the feedback is obtained on partial code, it concerns only isolated coding mistakes and avoids contextual clutter.
Further, the technique works with black-box models and reuses existing compiler infrastructure.

To enable generative compilation, the core new concept we introduce is a \emph{sealor}: a transformation that closes a partial program into a complete one by filling in missing syntactic structure and inserting well-typed placeholders.
The resulting sealed program can then be checked by an existing compiler.
The key challenge is designing the sealor to achieve both faithful feedback and error coverage.
Compiler rejection of the sealed program should guarantee that the partial program truly requires revision.
At the same time, sealing should expose as many dead ends as possible, so that unrecoverable errors are caught early.
We formalize these properties through corresponding completeness and soundness guarantees.

Our key insight is that a lightweight, mostly syntax-guided sealor can already meet the guarantees we target, without reimplementing the target language's type system.
We first demonstrate this on Featherweight Rust (\FR), a compact calculus inspired by Rust's ownership and borrowing discipline~\citep{pearce2021featherweight}.
We develop a sealor for \FR and prove that it satisfies the desired properties, with the full development mechanized in Lean.
We then carry the same methodology to real Rust: the sealor preserves generated syntax, inserts a small set of placeholders for missing code, and delegates type, borrow, and lifetime reasoning to \rustc{}, the Rust compiler.

We perform an extensive evaluation across seven frontier black-box and open-weight LLMs.
We focus on two repository-level code generation tasks for Rust, designed to stress compiler feedback: C-to-Rust translation and code generation against recently updated library APIs.
The results show that, compared with post-generation feedback, the on-the-fly feedback loop enabled by generative compilation further reduces compiler errors and improves functional correctness in most model-task configurations.
We then analyze how: generative compilation provides more focused diagnostics that are close to the error source and available well before file completion.

\paragraph{Key Contributions}
In summary, our main contributions are:
\begin{itemize}[leftmargin=*]
\item We introduce generative compilation, a novel framework enabling off-the-shelf compilers to check partial programs during code generation. We define the notion of sealors and prove that their completeness and soundness lift to generative compilation. (\cref{sec:framework})
\item We instantiate generative compilation on \FR, presenting a syntax-guided sealor that never rejects partial programs that admit valid completion and correctly flags an important class of partial programs with no valid extension. We fully mechanize \FR, our sealor, and its guarantees in Lean, making a number of corrections to the original \FR formalization in the process. (\cref{sec:lang,sec:method})
\item We apply our methodology to real Rust and build the first partial-program checker for Rust. This is based on a practical sealor that delegates semantic checking to \rustc{} while handling Rust-specific syntax, control flow, placeholders, and future-dependent errors. (\cref{sec:rust})
\item We evaluate our approach extensively on Rust coding tasks across frontier and open-weight models, showing that it reduces compiler errors, improves functional correctness, and detects various errors early during generation. (\cref{sec:eval})
\end{itemize}

%% file: sections-new/motivation.tex
\providecommand{\completes}{\rightsquigarrow}
\providecommand{\sealor}{\ensuremath{\mathcal{S}}}
\providecommand{\Compile}{\ensuremath{\mathcal{C}}}
\providecommand{\Pcompile}{\ensuremath{\mathcal{G}}}
\providecommand{\lang}{\ensuremath{\mathcal{L}}}
\providecommand{\ok}{\ensuremath{\mathsf{ok}}}
\providecommand{\err}{\ensuremath{\mathsf{err}}}
\providecommand{\Bool}{\ensuremath{\mathbb{B}}}
\providecommand{\vocab}{\ensuremath{\mathcal{V}}}
\providecommand{\alphabet}{\ensuremath{\Sigma}}
\providecommand{\EOS}{\ensuremath{\mathsf{EOS}}}
\providecommand{\valTrue}{\ensuremath{\mathrm{true}}}
\providecommand{\valFalse}{\ensuremath{\mathrm{false}}}
\providecommand{\concat}{\ensuremath{\circ}}
\providecommand{\emptystr}{\ensuremath{\texttt{\textquotedbl\textquotedbl}}}
\providecommand{\lm}{\ensuremath{M}}
\providecommand{\prefixpred}{\ensuremath{\mathcal{P}}}
\providecommand{\pclass}{\ensuremath{\mathcal{X}}}
\providecommand{\kw}[1]{\ensuremath{\mathtt{#1}}}
\providecommand{\missingterm}{\ensuremath{\kw{hole}_\kw{val}()}}

\section{Motivation for Generative Compilation}
\label{sec:motivation}

Today's AI-based code generation is largely powered by LLMs.
However, LLM-generated code provides no guarantee of respecting the target language's static rules.
Two existing strategies aim to address this gap: post-generation compiler feedback and constrained decoding.
In this section, we present both strategies and use a running example to show their limitations, thereby motivating generative compilation.

\paragraph{Code Generation with Language Models}
Let $\lm$ denote an LLM that generates code \emph{autoregressively}, one token at a time from left to right.
These tokens are drawn from a fixed vocabulary $\alphabet$, determined by $\lm$'s tokenizer.
Modern LLM tokenizers are universal: any string can be constructed from tokens in $\alphabet$~\citep{sennrich2016bpe}.
We therefore take $\alphabet^{\ast}$ to be the domain of all strings.

\begin{wrapfigure}{r}{0.48\linewidth}
\vspace{-5mm}
\begin{lstlisting}[basicstyle=\ttfamily\small, breaklines=true]
fn fmt_msg(msg: &mut String) -> String {
  let tag = &msg[..HEADER_SIZE];
  msg.drain(..HEADER_SIZE);
  let prio = tag.starts_with("HIGH");
  let body = msg;
  let n = body.lines().count();
  format!("{} {}: {}", prio, n, body)
}
\end{lstlisting}
\vspace{-4mm}
\caption{Non-compilable LLM-generated Rust code.}
\label{fig:bg-buggy}
\vspace{-5mm}
\end{wrapfigure}
Given a prompt $x$ (a natural-language task, optionally with reference code) and the code generated so far $c$, $\lm$ produces a next-token distribution $\lm(\cdot \mid x \concat c)$, where $\concat$ denotes string concatenation.
A token $t \sim \lm(\cdot \mid x \concat c)$ is sampled, appended to $c$, and generation repeats until a special end-of-sequence token $\EOS$ is produced.
As discussed above, this process can generate any string in $\alphabet^{\ast}$.
However, for a target language $\lang \subseteq \alphabet^{\ast}$, there is no \emph{a priori} guarantee that the final output lies in $\lang$.
For example, when prompted to write a Rust function that strips a fixed-size header from a message and formats the result for logging, an LLM might produce the code in \cref{fig:bg-buggy}.
Running the Rust compiler on this program reveals that it cannot be compiled.

\paragraph{Post-Generation Compiler Feedback}
Most programming languages provide a practical checker $\Compile$ that decides, for a given string $c$, whether it belongs to the language $\lang$, i.e., whether $c$ is syntactically well-formed and type checks.
Although compilers perform many tasks, such as optimizing code and emitting assembly, it is common to refer to this checking phase as ``compiling''; following this convention, we call the checker a ``compiler'' in this paper.

Formally, we model the compiler as a function $\Compile : \alphabet^{\ast} \to \Bool \times \alphabet^{\ast}$, which maps an input string $c$ to a pair $(\ok, \err)$ consisting of a Boolean verdict and a textual diagnostic.
If $c \in \lang$, then $\Compile(c) = (\valTrue, \emptystr)$, i.e., the textual diagnostic is empty;
otherwise, $\Compile(c) = (\valFalse, \err)$, where $\err \neq \emptystr$ is an error message explaining why $c \notin \lang$, e.g., by highlighting a failed type check, and potentially suggesting how to resolve the issue, e.g., by suggesting a type cast.
The compiler can be naturally used in a feedback loop for LLM-based code generation: after the model generates a full program $c$, we obtain $(\ok,\err) \coloneq \Compile(c)$.
If $\ok=\valTrue$, generation terminates and $c$ is guaranteed to be an element of $\lang$.
Otherwise, the model is asked to try again with the original prompt augmented by the generated program and diagnostic, e.g., $x \concat c \concat \err$.

\begin{figure*}[!t]
\centering
\begin{lstlisting}[language={}, basicstyle=\ttfamily\small, breaklines=true, numbers=none, escapechar=@, mathescape=false]
@\textcolor{red}{error[E0502]}@: @\textbf{cannot borrow `*msg` as mutable because it is also borrowed as immutable}@
2 |     let tag = &msg[..HEADER_SIZE];
  |                @\textcolor{violet}{\textendash\textendash\textendash\ immutable borrow occurs here}@
3 |     msg.drain(..HEADER_SIZE);
  |     @\textcolor{red}{\textasciicircum\textasciicircum\textasciicircum\textasciicircum\textasciicircum\textasciicircum\textasciicircum\textasciicircum\textasciicircum\textasciicircum\textasciicircum\textasciicircum\textasciicircum\textasciicircum\textasciicircum\textasciicircum\textasciicircum\textasciicircum\textasciicircum\textasciicircum\textasciicircum\textasciicircum\textasciicircum\textasciicircum\ mutable borrow occurs here}@
4 |     let prio = tag.starts_with("HIGH");
  |                @\textcolor{violet}{\textendash\textendash\textendash\ immutable borrow later used here}@
\end{lstlisting}
\vspace{-3mm}
\caption{Error message returned by the Rust compiler after it rejects the program in \cref{fig:bg-buggy}.}
\label{fig:bg-error}
\end{figure*}

For the example program in \cref{fig:bg-buggy}, running the compiler returns $(\valFalse, \err)$, with $\err$ shown in \Cref{fig:bg-error}.
The error involves Rust's ownership, borrowing, and lifetime discipline: a value may be accessed through either shared immutable references or an exclusive mutable reference, but not both during overlapping lifetimes.
At Line 2, the slice \texttt{tag} is an immutable borrow of \texttt{msg}.
The call to \texttt{msg.drain} at Line 3 then attempts to mutably borrow \texttt{msg} while that immutable borrow is still live, because \texttt{tag} is later used at Line 4.
This diagnostic is then appended to the prompt, and the model is asked to fix the error, either by regenerating the program from scratch or by editing the existing program.
\cref{fig:bg-fixed} shows a valid outcome that fixes the error by making \texttt{tag} an owned copy of the header rather than a borrow from \texttt{msg}.
The temporary immutable borrow used to create this copy ends before \texttt{msg.drain} takes a mutable borrow, so the two borrows do not overlap and the program satisfies Rust's rules.

\begin{wrapfigure}{r}{0.52\linewidth}
\vspace{-5mm}
\begin{lstlisting}[basicstyle=\ttfamily\small, breaklines=true]
fn fmt_msg(msg: &mut String) -> String {
  let tag = msg[..HEADER_SIZE].to_string();
  msg.drain(..HEADER_SIZE);
  let prio = tag.starts_with("HIGH");
  let body = msg;
  let n = body.lines().count();
  format!("{} {}: {}", prio, n, body)
}
\end{lstlisting}
\vspace{-4mm}
\caption{The regenerated Rust program above compiles after a post-generation compiler feedback round.}
\vspace{-5mm}
\label{fig:bg-fixed}
\end{wrapfigure}
This feedback loop pays for its simplicity through delayed feedback.
In \cref{fig:bg-buggy}, the borrow conflict already manifests at Line~4.
However, compiler feedback is obtained only after the model finishes generating the entire function at Line~8.
For longer programs in practice, early mistakes can influence subsequent generation and snowball before the model receives any signal from the compiler.
As a result, additional mistakes can be introduced and multiple errors may be reported together in a single diagnostic bundle, making each repair attempt harder and slowing the overall trial-and-error design process.

\paragraph{Constrained Decoding}
Constrained decoding offers an alternative to compiler feedback loops.
It exploits the autoregressive nature of LLM generation by intervening during decoding, rather than waiting until the entire program has been emitted. Specifically, constrained decoding relies on a \emph{prefix checker} $\prefixpred : \alphabet^{\ast} \to \Bool$, where $\prefixpred(c)=\valTrue$ if the current output $c$ can still be extended into a full program in $\lang$, or formally, if $\exists w.\ c \concat w \in \lang$. At each decoding step, candidate tokens are filtered through $\prefixpred$, so that only tokens $t$ satisfying $\prefixpred(c \concat t)=\valTrue$ may be selected. As a result, invalid continuations are ruled out as soon as they arise, and every generated partial program remains a prefix of some valid program in $\lang$.

Prefix checking solves a different problem from conventional compilation: the former must reason about partial programs, whereas the latter assumes full programs as input.
This creates a major implementation burden.
Existing compiler infrastructure cannot be used directly, so constrained decoding systems must often reimplement the relevant language semantics for prefix checking.
For general-purpose programming languages, matching the behavior of the full compiler in this way is impractical.
For example, the compiler-related part of the Rust codebase alone contains more than 600k lines of code~\citep{rustc}.
As a result, existing constrained decoding systems either enforce only syntax~\citep{park2024gad,dominos-DBLP:journals/corr/abs-2403-06988,DBLP:journals/tmlr/UgareSKM025} or cover a restricted subset of typing~\citep{ts,chopchop}.

\sethlcolor{red!35}
\begin{wrapfigure}{r}{0.49\linewidth}
\vspace{-5mm}
\begin{lstlisting}[basicstyle=\ttfamily\small, breaklines=true, escapechar=@]
fn fmt_msg(msg: &mut String) -> String {
  let tag = &msg[..HEADER_SIZE];
  msg.drain(..HEADER_SIZE);
  let prio = tag@\hl{.}@
\end{lstlisting}
\vspace{-4mm}
\caption{Constrained decoding rejects \hl{\texttt{.}} as it leads to a dead-end program prefix.}
\vspace{-5mm}
\label{fig:bg-partial}
\end{wrapfigure}
Assume for now that a (well-implemented) constrained decoding engine for Rust is available, 
and let us again examine the example program in \cref{fig:bg-buggy}.
Constrained decoding has no effect as long as each next token keeps the partial output extendable to a valid program.
Thus, the model generates the same prefix up to \texttt{let prio = tag}, shown in \cref{fig:bg-partial}.
This prefix is still valid: \texttt{tag} may refer to the prefix of an identifier to be defined later, such as a function name.
However, when selecting the next token, the \hl{\texttt{.}} used in the original program must be rejected by prefix checking.
Choosing \hl{\texttt{.}} commits \texttt{tag} to the existing variable, which triggers the same borrow conflict.
Therefore, to remain on a valid path, the model must choose the unnatural use-before-definition continuation.

This example exposes two further costs of constrained decoding.
First, constrained decoding silently filters out the token \hl{\texttt{.}} and can only shape future continuations.
The model is never told \emph{why} the token was rejected, nor can it revise the already generated prefix, as the repair in \cref{fig:bg-fixed} would require.
The only way forward is to force the model onto the low-probability use-before-definition path described above.
This can degrade global generation quality, especially when the rejection is due to non-local semantic constraints such as Rust's borrow checking~\citep{park2024gad, lipkin2025awrs, dominos-DBLP:journals/corr/abs-2403-06988}.
Second, filtering and resampling tokens require white-box access to the model's next-token distribution, which rules out the closed-weight frontier models that currently dominate coding workflows in practice.

\paragraph{Generative Compilation on the Running Example}
Prior work has made valuable progress on addressing the limitations of constrained decoding from the perspective of probabilistic sampling~\citep{dominos-DBLP:journals/corr/abs-2403-06988,DBLP:conf/iclr/UgareGS0M25,park2024gad,DBLP:journals/corr/abs-2506-05754,loula2025syntactic,lipkin2025awrs,DBLP:journals/corr/abs-2511-22277}.
In contrast, our generative compilation approach aims to explore a different route.
Like constrained decoding, it operates on partial programs rather than waiting for a full program.
Like conventional compilation, however, it produces textual feedback rather than altering the sampling process.
Crucially, it does so through a lightweight implementation that largely reuses existing compiler infrastructure and supports black-box models available only through APIs.
As a result, generative compilation addresses the limitations of the two strategies above.

We now use the running example to sketch how generative compilation handles partial programs. 
The key idea behind it is to seal partial programs, turning them into full programs that can be checked by a conventional compiler.
From the discussion above, the error in \cref{fig:bg-buggy} already manifests when generation reaches the prefix in \cref{fig:bg-partial}, including \hl{\texttt{.}}.
There, the program has committed to referring \texttt{tag} to the previously defined variable, which violates Rust's borrow-checking rules.
Through a series of rule applications, our sealor transforms this partial program into the full program shown in \cref{fig:bg-sealed}.
The sealed program preserves the access to \texttt{tag} at Line~4.
It then inserts a generic placeholder, $\missingterm$ at Line~5, which represents an arbitrary well-typed value.
In this example, it is a value of the function's \kw{String} return type.
Finally, the sealor closes the function at Line~6 such that the program is syntactically well-formed.

\begin{wrapfigure}{r}{0.49\linewidth}
\vspace{-5mm}
\begin{lstlisting}[basicstyle=\ttfamily\small, breaklines=true]
fn fmt_msg(msg: &mut String) -> String {
  let tag = &msg[..HEADER_SIZE];
  msg.drain(..HEADER_SIZE);
  &tag;
  |\missingterm{}|
}
\end{lstlisting}
\vspace{-4mm}
\caption{A full Rust program produced by our sealor.}
\vspace{-2mm}
\label{fig:bg-sealed}
\end{wrapfigure}
Compiling the sealed program exposes the same semantic issue as the full program in \cref{fig:bg-buggy}: \texttt{tag} still borrows from \texttt{msg} across the call to \texttt{msg.drain}.
Thus, \Rustc{} rejects \cref{fig:bg-sealed} with a similar borrow-checking diagnostic.
Compared with the partial program in \cref{fig:bg-partial}, the sealed program has only inserted $\missingterm$, which does not introduce type errors, and completed the missing syntax.
Thus, the borrow-checking error reported on the sealed program reflects a genuine error already committed by the partial program.
This is one of the guarantees we aim to obtain from the sealor: compiler errors on sealed programs should correspond to errors in the original partial programs, rather than artifacts of sealing.

We then map the diagnostic produced on the sealed program back to the partial program by comparing source-code spans.
The model receives the partial program together with the mapped diagnostic, but never sees the sealed program itself.
With this feedback, it ultimately produces the same repaired program as in \cref{fig:bg-fixed}, making \texttt{tag} an owned copy rather than a borrow from \texttt{msg}.
Generative compilation therefore provides the relevant feedback four and a half lines earlier than post-generation feedback, when the error first becomes unavoidable.

%% file: sections-new/gencomp.tex
\section{Generative Compilation}
\label{sec:framework}
Next, we formally define generative compilation in a generic, language-agnostic form.

\subsection{Generative Compilers and Sealors}
\label{sec:defs}

\paragraph{Generative Compilers: Definition}
A \emph{generative compiler} is a function $\Pcompile : \alphabet^{\ast} \to \Bool \times \alphabet^{\ast}$.
That is, it has the same interface as a conventional compiler and returns a Boolean verdict $\ok$ together with a textual diagnostic $\err$.
However, instead of checking whether a complete program belongs to $\lang$, it performs prefix checking, i.e., whether a partial program $c$ admits a valid completion.
Unlike constrained decoding, a generative compiler retains compiler diagnostics: $\err$ is empty whenever $\ok=\valTrue$ and otherwise provides a textual explanation, rather than a silent rejection signal.

\paragraph{Generative Compilers: Completeness and Soundness}
We say that a generative compiler $\Pcompile$ \emph{accepts} $c$ when it returns a positive verdict on $c$, i.e., when $\Pcompile(c) = (\valTrue, \emptystr)$.
It is \emph{complete} on a class of programs $\pclass \subseteq \alphabet^{\ast}$ if for all $c \in \pclass$, $\exists w.\ c \concat w \in \lang$ implies that $\Pcompile$ accepts $c$.
It is \emph{sound} on $\pclass$ if for all $c \in \pclass$, $\Pcompile$ accepts $c$ implies $\exists w.\ c \concat w \in \lang$.
Intuitively, completeness rules out rejecting partial programs that admit valid extensions, whereas soundness rules out accepting dead-end partial programs.
We say $\Pcompile$ is \emph{globally} complete or sound, when $\pclass = \alphabet^{\ast}$.
$\Pcompile$ is \emph{exact} on $\pclass$ iff it is sound and complete on $\pclass$.
Note that the first output of $\Pcompile(c)$ can be viewed as a prefix checker, inheriting the corresponding notions of soundness and completeness \citep{ts}.

\paragraph{Sealors: Definition}
Generative compilation seeks to reuse the existing compiler to obtain rich feedback and avoid costly reimplementation, while checking partial programs during generation.
This is achieved through a new notion of \emph{sealing}, which we define now.
A \emph{sealor} is a lightweight, language-specific transformation $\sealor: \alphabet^{\ast} \to \alphabet^{\ast}$ that attempts to close the unfinished structure of a partial program.
Its goal is to produce a closed form that a conventional compiler can process, thereby obtaining a verdict and diagnostic.

Given a compiler $\Compile$ and a sealor $\sealor$, the induced generative compiler is defined by $\Pcompile_{\Compile,\sealor}(c) \coloneq \Compile(\sealor(c))$.
This definition abstracts away a practical mismatch: diagnostics produced by $\Compile$ refer to the sealed program $\sealor(c)$, whereas the model should receive feedback on the original partial program $c$.
This can be addressed by post-processing diagnostics, for example, by mapping source spans in $\sealor(c)$ back to their corresponding spans in $c$.
We discuss this implementation detail for Rust in \cref{sec:rust-distinctions}; for simplicity, we assume for now that diagnostics are reported on the corresponding spans of the partial program $c$.

\paragraph{Sealors: Completeness and Soundness}
We define completeness and soundness for sealors analogously to generative compilers.

\begin{definition}[Sealor Completeness and Soundness]\label{def:sealor-sound-complete}
A sealor $\sealor$ is \emph{complete} on $\pclass \subseteq \alphabet^{\ast}$ if, for every $c \in \pclass$, $\exists w.\ c \concat w \in \lang$ implies $\sealor(c) \in \lang$. It is \emph{sound} on $\pclass$ if, for every $c \in \pclass$, $\sealor(c) \in \lang$ implies $\exists w.\ c \concat w \in \lang$. It is \emph{globally} complete (resp.\ sound) when $\pclass = \alphabet^{\ast}$.
\end{definition}

Intuitively, completeness requires that sealing turns every prefix with a valid continuation into a valid program, whereas soundness requires that sealing not make a dead-end prefix appear valid.
Assuming $\Compile$ is exact, the completeness and soundness of a sealor correspond directly to those of its induced generative compiler.

\begin{theorem}[Sealor Completeness and Soundness lift to the Generative Compiler]\label{thm:cons-prefix}
Let $\Compile$ be exact and $\pclass \subseteq \alphabet^{\ast}$. If $\sealor$ is complete on $\pclass$, then $\Pcompile_{\Compile,\sealor}$ is complete on $\pclass$. Symmetrically, if $\sealor$ is sound on $\pclass$, then $\Pcompile_{\Compile,\sealor}$ is sound on $\pclass$.
\end{theorem}
\begin{proof}
For completeness: take $c \in \pclass$ such that $\exists w.\ c \concat w \in \lang$. Completeness of $\sealor$ on $\pclass$ gives $\sealor(c) \in \lang$; exactness of $\Compile$ then gives that $\Compile$ accepts $\sealor(c)$, i.e., that $\Pcompile_{\Compile,\sealor}$ accepts $c$.
For soundness: take $c \in \pclass$ with $\Pcompile_{\Compile,\sealor}$ accepting $c$, i.e., $\Compile$ accepts $\sealor(c)$. Exactness of $\Compile$ gives $\sealor(c) \in \lang$; soundness of $\sealor$ on $\pclass$ then gives $\exists w.\ c \concat w \in \lang$.
\end{proof}

These characterizations give semantic proof obligations for a sealor, which can be established from the metatheory of the target language and then lifted to completeness or soundness of the induced generative compiler.
We dive deep into this for our \FR sealor in \cref{sec:method}.

\paragraph{Target Guarantees: Global Completeness and Selective Soundness}
Achieving both global completeness and global soundness amounts to exact prefix checking, which is undecidable in general for expressive semantic constraints~\citep{chopchop}.
Constrained decoding faces this difficulty at every decoding step.
Because emitted tokens cannot be revisited, it typically requires global soundness to inductively ensure that generation always remains on a path that admits a valid extension.
Token-level completeness is not required, but highly desirable: rejecting a token prematurely that could still lead to a valid program may significantly degrade model performance under constraints \citep{dominos-DBLP:journals/corr/abs-2403-06988,tam2024letspeakfreelystudy}.
Maintaining both properties token by token pushes constrained decoding toward exact prefix checking.
When constrained decoding should support rich static semantics, this causes substantial reimplementation effort of the corresponding semantic analysis such as typing \citep{ts}.

For generative compilation, the priorities for soundness and completeness change, with incompleteness becoming the more harmful failure mode.
Rejecting a partial program that admits a valid continuation causes the model to discard or revise a prefix from which a correct program could still be generated, providing spurious feedback that can confuse the model.
Meanwhile, soundness can be relaxed.
A rejection triggers revision: the model regenerates based on diagnostic feedback, allowing to revise already generated tokens.
Missing a soundness violation is thus tolerable: the generative compiler may accept a dead-end prefix and allow generation to continue, but the error can still be caught by a later check or by final compilation, when more code context is available.
Nevertheless, soundness remains valuable where it can be obtained cheaply, because it lets on-the-fly feedback arrive as early as possible.
Therefore, for generative compilation, we target global completeness together with soundness on a selected class of programs $\pclass \subsetneq \alphabet^{\ast}$.

\begin{figure}[t]
  \centering
  \begin{tikzpicture}[
    module/.style={draw, semithick, rounded corners=6pt, inner sep=10pt, align=left, minimum height=3.3cm, minimum width=3.25cm},
    title/.style={font=\small\bfseries, align=center},
    arr/.style={-{Latex[length=3mm]}, semithick},
    lbl/.style={font=\footnotesize, fill=white, inner sep=2pt},
  ]
    \node[module, fill=blue!5] (llm) at (0,0) {\begin{minipage}{3.9cm}
      \raggedright\small
      \textbf{whenever} a prompt $x$ arrives:\\
      \hspace{1em} $c \gets \emptystr$\\
      \hspace{1em} \textbf{repeat:}\\
      \hspace{2em} $t \sim \lm(\cdot \mid x \concat c)$\\
      \hspace{2em} $c \gets c \concat t$\\
      \hspace{2em} send $c$\\
      \hspace{1em} \textbf{until} $t = \EOS$
    \end{minipage}};
    \node[title, above=1.3mm of llm] {LLM-Based Code Generation};

    \node[module, fill=orange!5, right=2.5cm of llm] (gc) {\begin{minipage}{3.9cm}
      \raggedright\small
      \textbf{loop}:\\
      \hspace{1em} $c \gets$ latest received prefix\\
      \hspace{1em} $(\ok, \err) \gets \Pcompile_{\Compile,\sealor}(c)$\\
      \hspace{1em} \textbf{if} $\neg\ok$:\\
      \hspace{2em} $x \gets x \concat c \concat \err$\\
      \hspace{2em} send $x$ \\
      \hspace{1em} \textbf{if} $\ok$ \textbf{and} $c$ ends with $\EOS$:\\
      \hspace{2em} \textbf{return} $c$
    \end{minipage}};
    \node[title, above=0.6mm of gc] {Generative Compilation};

    \draw[arr] ([yshift=1cm]llm.east) --
      node[lbl, above]{prefix $c$}
      ([yshift=1cm]gc.west);
    \draw[arr] ([yshift=-1cm]gc.west) --
      node[lbl, below]{prompt $x$}
      ([yshift=-1cm]llm.east);
  \end{tikzpicture}
  \vspace{-2mm}
  \caption{We combine generative compilation and LLM-based code generation as two (concurrent) modules. The LLM streams each generated prefix $c$ to the generative compiler; upon rejection, the generative compiler returns an augmented prompt containing $c$ and its diagnostic $\err$, causing generation to restart. The two modules communicate only through plain text, supporting black-box LLM APIs.}
  \label{fig:gc-threads}
\end{figure}
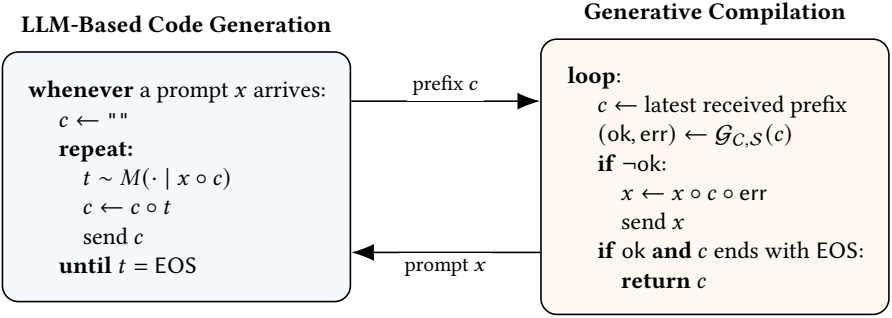

\subsection{Combining Generative Compilation and LLM-Based Code Generation}
\label{sec:decoding}

\cref{fig:gc-threads} shows the approach we use in this paper to combine generative compilation with LLM-based code generation: the two run as (concurrent) modules that exchange only plain text.
Other integration strategies are possible; we discuss them in \cref{sec:dis}.

The module on the left represents LLM-based code generation.
It performs standard autoregressive decoding, as discussed in \cref{sec:motivation}.
Generation is interrupted and restarted whenever the module receives a new prompt $x$ from generative compilation.
As token generation proceeds, the module sends the latest partial program $c$ to generative compilation, stopping when $\EOS$ is emitted.
This module does not require access to $\lm$'s token distribution or control over its sampling procedure.
It only observes the current partial program $c$ as it is generated.
For frontier black-box models, this can be obtained through streaming API responses, where output is returned incrementally rather than only after completion.
Although \cref{fig:gc-threads} illustrates communication at token granularity, the same design also works when the API returns larger chunks of text.

The module on the right performs generative compilation.
It checks received prefixes using a \emph{latest-wins} strategy: incoming prefixes do not interrupt ongoing validation, and prefixes received during validation are coalesced so that only the newest one is retained.
When validation of a prefix $c$ finishes with $(\ok,\err) \coloneq \Pcompile_{\Compile,\sealor}(c)$, the module acts on the verdict.
If $\neg\ok$, it augments the prompt with $c$ and $\err$ and sends the revised prompt to the generator.
If $\ok$ holds and $c$ ends with $\EOS$, it returns $c$ as the generated program.
Otherwise, $c$ is a valid partial program, and the module begins validating the latest buffered prefix, if one is available.
In the common case where the $\EOS$-terminated prefix is already a complete program, the final check coincides with ordinary post-generation compiler feedback.
More generally, generative compilation and post-generation feedback are complementary and can always be combined.

When configured to run concurrently, generative compilation does not block token generation: the LLM can continue producing code while one prefix is being validated.
This avoids waiting for a compiler result after every token when validation is slower than token sampling.
The latest-wins strategy then retains only the newest prefix generated during validation for the next check.
As a result, a rejection may arrive after the LLM has produced a few additional tokens.
This extra work is modest and limited to the short suffix generated during a single validation, and once the earlier prefix is rejected, its continuation is discarded in the restarted attempt.
When incremental compilation is available, as in \rustc{}, compilation speed can increase significantly.
Generative compilation benefits also from this, since it is built largely on existing compiler infrastructure.

%% file: sections-new/lang.tex
\section{\FR: A Core Calculus for Rust}
\label{sec:lang}


\providecommand{\kw}[1]{\ensuremath{\mathtt{#1}}}
\newcommand{\sep}[1]{\ensuremath{\mathbin{|}}}

\newcommand{\term}{\ensuremath{t}}
\newcommand{\lval}{\ensuremath{w}}
\newcommand{\val}{\ensuremath{v}}
\newcommand{\pval}{\ensuremath{v^{\bot}}}
\newcommand{\typ}{\ensuremath{T}}
\newcommand{\ptyp}{\ensuremath{\tilde{T}}}
\newcommand{\litval}{\ensuremath{c}}
\newcommand{\lft}{\ensuremath{l}}

\newcommand{\ownref}{\ensuremath{\ell^{\bullet}}}
\newcommand{\borref}{\ensuremath{\ell^{\circ}}}

\newcommand{\bx}{\ensuremath{\square}}

\newcommand{\bshr}{\ensuremath{\kw{\&}}}
\newcommand{\bmut}{\ensuremath{\kw{\&}\,\kw{mut}}}

\newcommand{\undefT}[1]{\ensuremath{\lfloor #1 \rfloor}}

\newcommand{\varctx}{\ensuremath{\Gamma}}              
\newcommand{\storectx}{\ensuremath{\sigma}}            

\newcommand{\copyfn}{\mathrm{copy}}                    
\newcommand{\readProh}{\mathrm{readProhibited}}        
\newcommand{\writeProh}{\mathrm{writeProhibited}}      
\newcommand{\mutpred}{\mathrm{mut}}                    
\newcommand{\movefn}{\mathrm{move}}                    
\newcommand{\writefn}[1]{\mathrm{write}^{#1}}          
\newcommand{\dropfn}{\mathrm{drop}}                    

\newcommand{\shapecompat}{\sim}                     
\newcommand{\outlives}{\succeq}                        

\newcommand{\slot}[2]{\ensuremath{\langle #1 \rangle^{#2}}}

\providecommand{\dgr}[1]{\colorbox{gray!25}{$#1$}}

\providecommand{\LifetimeChild}{\mathrel{\succeq{}_{\text{imm}}}}

In this section, we review Featherweight Rust~\citep{pearce2021featherweight}, a compact calculus inspired by Rust's borrow system.
We present the syntax (\cref{sec:lang-syntax}) and typing (\cref{sec:lang-typing}) needed to demonstrate our generative compilation approach in \cref{sec:method}, omitting details that are not essential here.
For the full calculus, we refer readers to~\citet{pearce2021featherweight}.

\paragraph{Why a Core Calculus}
To present and analyze our approach formally, we need a language formalization.
Since our target is Rust, this means a formalization for Rust, and it must satisfy two requirements.
First, it should remain close to Rust's surface syntax, because our approach directly manipulates Rust source code.
Second, it should provide a compact setting in which our approach can be defined, analyzed, and understood, while still capturing Rust's most prominent features.

Existing Rust formalizations do not simultaneously meet these requirements.
RustBelt~\citep{rustbelt}, RustHorn~\citep{rusthorn}, and Aeneas~\citep{aeneas} provide powerful foundations for semantic reasoning and verification, but they operate at a level of abstraction that is not sufficiently close to the source syntax, and their developments are necessarily substantial.
Oxide~\citep{oxide} is closer to Rust's surface syntax, but still introduces significant machinery, making it more elaborate than needed for our purpose.

\paragraph{Why \FR}
Featherweight Rust (\FR)~\cite{pearce2021featherweight} follows the philosophy of Featherweight Java~\cite{igarashi2001featherweight}: provide a concise and minimal proof of type soundness while preserving the essence of the soundness argument for the full language.
Concretely, \FR remains close to Rust's surface syntax and captures many key aspects of Rust, including copy and move semantics, mutable and immutable borrows, and lexical lifetimes.
Moreover, \FR is compact and establishes soundness guarantees.

\input{figures/grammar-fr.tex}

\subsection{Syntax}
\label{sec:lang-syntax}

\FR's syntax is shown entirely in \cref{fig:grammar}.
We highlight some important details in order of appearance.

\paragraph{Move vs.\ Copy is Syntactic}
In Rust, a use of $\lval$ is a copy if its type implements the $\kw{Copy}$ trait, and a move otherwise.
\FR instead makes this distinction syntactically: $\lval$ denotes a move, whereas $\kw{copy}\ \lval$ denotes a copy.
This design keeps the operational semantics independent of the type system.
The type system (defined in \cref{sec:lang-typing}) then enforces that copies are only performed on copyable values: $\kw{copy}\ \lval$ is well typed only if $\lval$ has type $\typ$ satisfying $\copyfn(\typ)$, \FR's analog of the $\kw{Copy}$ trait.
In our mechanization, the predicate holds for $\epsilon$, $\kw{int}$, and shared borrows, but not for mutable borrows or boxes. Treating unit as copyable is a small adaptation relative to~\citet{pearce2021featherweight}.

\paragraph{Lifetimes are Lexical}
A block $\{^{\lft}\, \overline{\term}\}$ is a sequence of terms separated by semicolons and annotated by a lifetime.
Lifetimes $\lft, m$ annotate every lexical block $\{^{\lft}\, \overline{\term}\}$.
Lifetimes form a partial order $\lft \succeq m$, which means that $m$ is inside $\lft$.
This holds exactly when $m$'s block is lexically nested inside $\lft$'s, e.g., $\{^{\lft}\, \{^{m}\, \overline{\term}\}\}$.
We note that modern Rust uses non-lexical lifetimes (NLL), which let a borrow end at its last use rather than at the enclosing block's exit.
\FR's lexical model admits a substantially simpler metatheory while capturing the essence of lifetimes and remaining sound.

The block's lifetime is written after the opening brace, $\{^{\lft}\, \overline{\term}\}$, rather than after the closing brace as in~\citet{pearce2021featherweight} (i.e., $\{\overline{\term}\}^{\lft}$).
The two notations denote the same abstract syntax, but placing the lifetime up front makes it visible that the label is committed at block-entry rather than at block-exit.
This matters for several typing rules in \cref{sec:lang-typing},
which type body terms against the block's ambient lifetime: the label must be fixed before the first body term is checked. It also matters for our partial syntax in \cref{sec:method}, where a partial block whose tail is being emitted must already carry its lifetime.


\paragraph{Reference Values Encode Drop Responsibility}
Reference values come in two flavors: an owning reference $\ownref$ and a borrowed reference $\borref$.
Dropping an $\ownref$ recursively drops the location it refers to, whereas dropping a $\borref$ does not.
Encoding this distinction at the value level keeps the operational semantics independent of the type system.
By contrast, the distinction between shared and mutable references appears only in the type system, not at runtime.
Both shared and mutable borrows therefore evaluate to $\borref$ values, since they are non-owning references.

\paragraph{Borrow Types Name the Borrowed LVal}
A reference type is either a shared borrow $\bshr\,\lval$ or a mutable borrow $\bmut\,\lval$; we write $\bshr\,[\kw{mut}]\,\lval$ if either mode is allowed, where the target is a single lval (e.g., $\bshr\,x$).
Naming the borrowed lval directly in the type allows validating borrow safety by tracking overlaps and conflicts between borrowed and assigned lvals.
In contrast to \citet{pearce2021featherweight}, we adopt \FR's follow-up work \citep{followup}, restricting the $\lval$ to a single lval instead of a set of lvals to establish cycle-freeness of borrows.
This formulation is sufficient in the absence of conditionals.

\paragraph{Partial Types and Values Track Moved-Out Storage}
At runtime, a partial value $\pval$ extends $\val$ with the undefined marker $\bot$, which marks a slot whose contents have been moved out.
At the type level, a partial type $\ptyp$ extends $\typ$ with $\bx\,\ptyp$ (a box whose contents are themselves partial) and $\undefT{\typ}$ (a slot moved out but with its former shape remembered, e.g., $\undefT{\bshr\,x}$, so re-assignment can be shape-checked).
Every value is trivially a partial value, and every type a partial type.
The converses fail because only the partial forms admit moved-out markers.

\paragraph{Source-Level Fragment}
The reference values $\ownref$ and $\borref$ arise only at runtime, as the results of evaluating $\kw{box}$ and borrow expressions; a programmer never writes them directly.
We call a term \emph{source-level} when no reference value appears as a subterm.
The language does not require a distinguished surrounding block: a top-level program is any term typable from the empty environments at the root lifetime $\lft_{\mathrm{root}}$.
Blocks remain the construct that introduces nested lexical lifetimes.


\input{figures/typing-fr.tex}

\subsection{Typing}
\label{sec:lang-typing}

\cref{fig:typing} presents the main typing rules of \FR that we use in later sections.
Next, we start by introducing the typing contexts and judgments, followed by the auxiliary functions, and then explain the rules. We close by stating \FR's type and borrow safety and summarizing our Lean mechanization.

\paragraph{Contexts and Judgments}
Typing uses two contexts.
The \emph{typing environment} $\varctx$ maps each declared variable $x$ to a slot type $\slot{\ptyp}{m}$, recording its partial type $\ptyp$ and the lifetime $m$ in which $x$ was declared.
The \emph{store typing} $\storectx$ assigns types to runtime locations.
$\storectx$ is needed only when typing intermediate execution states; source-level programs are typed under $\storectx = \emptyset$.

The typing rules in \cref{fig:typing} use one main judgment: $\varctx_1 \vdash \langle \term : \typ \rangle^{\lft}_{\storectx} \dashv \varctx_2$.
It asserts that, in lifetime $\lft$, the term $\term$ has type $\typ$, and that evaluating $\term$ transforms the environment from $\varctx_1$ to $\varctx_2$. Threading the environment through the judgment in this input/output style makes the type system \emph{flow-sensitive}.
For example, after moving from an lval $\lval$, the corresponding slot in $\varctx_2$ is marked moved out, so a subsequent use of $\lval$ is rejected.
Declarations and assignments similarly update the environment in a way that affects subsequent typing decisions.
These effects cannot be captured by a traditional flow-insensitive type system, in which each variable has a single fixed type throughout its scope.

\cref{fig:typing} uses three auxiliary judgments only as premises of the main rules: LVal typing, type well-formedness, and shape compatibility.
\emph{LVal typing}, $\varctx \vdash \lval : \slot{\ptyp}{m}$, determines the partial type of $\lval$ and the lifetime $m$ in which it can safely be used.
\emph{Type well-formedness}, $\varctx \vdash \typ \outlives \lft$, holds when every borrow in $\typ$ targets storage that outlives $\lft$.
\emph{Shape compatibility}, $\varctx \vdash \ptyp_1 \shapecompat \ptyp_2$, holds when two partial types agree in structural shape modulo lifetimes and undefined components, while still comparing the types beneath any undefined components.
Each of these three has routine rules, so we defer to \citep{pearce2021featherweight} for their full definition.

\paragraph{Auxiliary Functions}
The typing rules also use \FR's standard auxiliary functions~\citep{pearce2021featherweight}.
Each is defined by straightforward case analysis over its input so we summarize their roles rather than reproduce the full definitions.
Several auxiliary predicates characterize accessibility.
$\copyfn(\typ)$ holds when $\typ$ admits copy semantics.
$\readProh(\varctx, \lval)$ holds when $\lval$ is currently mutably borrowed in $\varctx$, whereas $\writeProh(\varctx, \lval)$ holds when it is borrowed in either mode.
$\mutpred(\varctx, \lval)$ holds when $\lval$ is mutably accessible: no dereference in $\lval$ steps through an immutable borrow.

Three update functions modify the environment.
$\movefn(\varctx, \lval)$ marks the value at $\lval$ as moved out by replacing the corresponding component of its slot type with $\undefT{\typ}$, where $\typ$ is the component's former type.
$\writefn{0}(\varctx, \lval, \typ)$ updates $\varctx$ to reflect assigning a value of type $\typ$ through $\lval$.
$\dropfn(\varctx, m)$ removes all bindings declared in lifetime $m$.

\paragraph{Typing Rules for Terms}
Each rule in \cref{fig:typing} types one corresponding term form from \cref{fig:grammar}.
\textsc{T-Const} checks literals and runtime values against $\storectx$ and leaves the environment unchanged.
\textsc{T-Box} assigns $\kw{box}\ \term$ the type $\bx\ \typ$ when $\term$ has type $\typ$, while preserving any environment changes caused by evaluating $\term$.
\textsc{T-Copy} and \textsc{T-Move} encode Rust's familiar distinction between copy and move semantics.
\textsc{T-Copy} permits copying an lval only when its type satisfies $\copyfn$ and $\readProh$ is false.
\textsc{T-Move} returns the lval's type, requires $\writeProh$ to be false, and marks the corresponding component of its slot as moved out via $\movefn$.

\textsc{T-ImmBorrow} and \textsc{T-MutBorrow} simulate Rust's distinction between shared and mutable borrows, establishing their mutual exclusion.
\textsc{T-ImmBorrow} creates a shared borrow only when $\lval$ is not currently \emph{mutably} borrows, i.e., $\readProh(\varctx,\lval)$ is false.
This permits multiple shared borrows while excluding simultaneous mutable borrows.
\textsc{T-MutBorrow} creates a mutable borrow only when $\lval$ is not currently borrowed in \emph{any} way, i.e., $\writeProh(\varctx,\lval)$ is false.
This reflects Rust's requirement that mutable access excludes any simultaneous borrows.
Additionally, $\lval$ must be mutable.
The borrow expression itself leaves the environment unchanged.
Its borrow information is recorded only when the resulting reference is stored through a declaration or assignment.

\textsc{T-Seq} threads the environment through a sequence of terms from left to right and gives the sequence the type of its last term.
\textsc{T-Block} types $\{^{m}\, \overline{\term}\}$ in its own lifetime $m$.
Its result type must be well formed in the enclosing lifetime $\lft$, which prevents borrows of block-local storage from escaping. After exiting the block, $\dropfn$ removes any bindings declared in $m$.
The premise $\lft \LifetimeChild m$, requiring $m$ to be the immediate child of $\lft$, corrects a missing assumption in \FR's typing that $m$ is inside $\lft$. We specify this as an immediate-child relationship, as one can always choose an immediate child $m$ in place of any assigned $m'$ inside $\lft$.

\textsc{T-Declare} first types the initializer and then extends the resulting environment with a fresh binding $x \mapsto \slot{\typ}{\lft}$. The freshness check is made in the post-initializer environment $\varctx_2$, rather than in $\varctx_1$ by~\citet{pearce2021featherweight}: otherwise an initializer can redeclare $x$ and return it to the environment before the outer declaration installs its binding.
\textsc{T-Assign} types the right-hand side to obtain $\varctx_2$ and $\typ_2$, then resolves the left-hand side in $\varctx_2$.
The lval-typing environment $\varctx_2$ is another correction to \FR's typing, since evaluating the right-hand side may re-type the path under $\lval$ (e.g., $*p = \{p;mut\,w\}$ moves $p$ out, thus making access to $*p$ invalid). The rule then requires $\typ_2$ to be shape-compatible with the destination's stored partial type and well formed for the destination's declaration lifetime $m$, updates the environment via $\writefn{0}$, and requires $\writeProh(\varctx_3,\lval)$ to be false.

\paragraph{Type and Borrow Safety}
\label{sec:lang-soundness}
\citet{pearce2021featherweight} establish \emph{type and borrow safety} for \FR: a well-typed program does not get stuck, and if it terminates, it does so in a store that respects the borrow invariant, i.e., free of dangling or duplicated owning references and with at most one mutable borrow per location. We leverage this result to treat \FR as a meaningful proxy for a usefully complex and expressive type system, motivating our choice of \FR.

\paragraph{Mechanization and Corrections}
We fully mechanize \FR in Lean, including syntax, typing, operational semantics, and type and borrow safety proofs.
In this process, we corrected the stated typing rules in several places.
The main corrections to \textsc{T-Block}, \textsc{T-Declare}, and \textsc{T-Assign} are marked with grey boxes in \cref{fig:typing}.
Beyond these typing adaptations, we correct the rules for unit values, align the runtime semantics with the reference implementation by \citet{pearce2021featherweight}, and strengthen the progress invariant with borrow safety and the linearizability property of \citet{followup}.
We refer to our extended writeup alongside our mechanization.

%% file: figures/grammar-fr.tex
\begin{figure}[tp]
  \small
  \centering
  \begin{minipage}{\textwidth}
    \begin{minipage}[t]{0.45\textwidth}
      \centering
      \begin{tabular}{l@{\hspace{5mm}}l}
$\term ::= $                                                                                       & Term \\
\quad $\val$                                                                                       & \quad Value \\
\quad $\kw{box}\ \term$                                                                            & \quad Heap allocation \\
\quad $\kw{copy}\ \lval$                                                                           & \quad Copy \\
\quad $\lval$                                                                                      & \quad Move \\
\quad $\bshr\,[\kw{mut}]\ \lval$                                                                   & \quad (Mutable) borrow \\
\quad $\{^{\lft}\, \overline{\term}\}$                                                             & \quad Block \\
\quad $\kw{let}\ \kw{mut}\ x\, =\, \term$                                                          & \quad Declaration \\
\quad $\lval\, =\, \term$                                                                          & \quad Assignment \\
\\
$\lval ::= $                                                                                       & LVal \\
\quad $x$                                                                                          & \quad Variable \\
\quad $*\lval$                                                                                     & \quad Dereference \\
\\
$\pval ::= $                                                                                       & Partial value \\
\quad $\val$                                                                                       & \quad Value \\
\quad $\bot$                                                                                       & \quad Undefined \\
      \end{tabular}
    \end{minipage}
    \hspace{6mm}
    \begin{minipage}[t]{0.05\textwidth}
      \begin{tabular}{l|l}
& \\ & \\ & \\ & \\ & \\ & \\ & \\ & \\ & \\
& \\ & \\ & \\ & \\ & \\ & \\ & \\ & \\
      \end{tabular}
    \end{minipage}
    \hspace{4mm}
    \begin{minipage}[t]{0.40\textwidth}
      \centering
      \begin{tabular}{l@{\hspace{4mm}}l}
$\val ::= $                                                                                        & Value \\
\quad $\epsilon$                                                                                   & \quad Unit value \\
\quad $\litval$                                                                                    & \quad Integer constant \\
\quad $\ownref$                                                                                    & \quad Owning reference \\
\quad $\borref$                                                                                    & \quad Borrowed reference \\
\\
$\ptyp ::= $                                                                                       & Partial type \\
\quad $\typ$                                                                                       & \quad Type \\
\quad $\bx\,\ptyp$                                                                                 & \quad Partial box \\
\quad $\undefT{\typ}$                                                                              & \quad Moved-out slot \\
\\
$\typ ::= $                                                                                        & Type \\
\quad $\epsilon$                                                                                   & \quad Unit type \\
\quad $\kw{int}$                                                                                   & \quad Integer type \\
\quad $\bmut\ \lval$                                                                               & \quad Mutable borrow \\
\quad $\bshr\ \lval$                                                                               & \quad Shared borrow \\
\quad $\bx\,\typ$                                                                                  & \quad Box type \\
      \end{tabular}
    \end{minipage}
  \end{minipage}
  \vspace{-2mm}
  \caption{\FR's Syntax, following \citet{pearce2021featherweight}.}
  \label{fig:grammar}
\end{figure}

%% file: figures/typing-fr.tex
\begin{figure}
  \small
  \centering
  \begin{equation*}
  \begin{array}{r@{\qquad}l}
    [\textsc{T-Const}]
    \inference{
      \storectx \vdash \val : \typ
    }{
      \varctx \vdash \langle \val : \typ \rangle^{\lft}_{\storectx} \dashv \varctx
    }
    &
    [\textsc{T-Box}]
    \inference{
      \varctx_1 \vdash \langle \term : \typ \rangle^{\lft}_{\storectx} \dashv \varctx_2
    }{
      \varctx_1 \vdash \langle \kw{box}\ \term : \bx\,\typ \rangle^{\lft}_{\storectx} \dashv \varctx_2
    }
    \\[6mm]
    [\textsc{T-Copy}]
    \inference{
      \varctx \vdash \lval : \slot{\typ}{m}
      \quad
      \copyfn(\typ)
      \\
      \neg\readProh(\varctx, \lval)
    }{
      \varctx \vdash \langle \kw{copy}\ \lval : \typ \rangle^{\lft}_{\storectx} \dashv \varctx
    }
    &
    [\textsc{T-Move}]
    \inference{
      \varctx_1 \vdash \lval : \slot{\typ}{m}
      \quad
      \varctx_2 = \movefn(\varctx_1, \lval)
      \\
      \neg\writeProh(\varctx_1, \lval)
    }{
      \varctx_1 \vdash \langle \lval : \typ \rangle^{\lft}_{\storectx} \dashv \varctx_2
    }
    \\[6mm]
    [\textsc{T-ImmBorrow}]
    \inference{
      \varctx \vdash \lval : \slot{\typ}{m}
      \\
      \neg\readProh(\varctx, \lval)
    }{
      \varctx \vdash \langle \bshr\ \lval : \bshr\ \lval \rangle^{\lft}_{\storectx} \dashv \varctx
    }
    &
    [\textsc{T-MutBorrow}]
    \inference{
      \varctx \vdash \lval : \slot{\typ}{m}
      \quad
      \mutpred(\varctx, \lval)
      \\
      \neg\writeProh(\varctx, \lval)
    }{
      \varctx \vdash \langle \bmut\ \lval : \bmut\ \lval \rangle^{\lft}_{\storectx} \dashv \varctx
    }
    \\[6mm]
    [\textsc{T-Seq}]
    \inference{
      \varctx_1 \vdash \langle \term_1 : \typ_1 \rangle^{\lft}_{\storectx} \dashv \varctx_2
      \\
      \cdots
      \\
      \varctx_n \vdash \langle \term_n : \typ_n \rangle^{\lft}_{\storectx} \dashv \varctx_{n+1}
    }{
      \varctx_1 \vdash \langle \overline{\term} : \typ_n \rangle^{\lft}_{\storectx} \dashv \varctx_{n+1}
    }
    &
    [\textsc{T-Block}]
    \inference{
      \dgr{\lft \LifetimeChild m}
      \quad
      \varctx_1 \vdash \langle \overline{\term} : \typ \rangle^{m}_{\storectx} \dashv \varctx_2
      \\
      \varctx_2 \vdash \typ \outlives \lft
      \quad
      \varctx_3 = \dropfn(\varctx_2, m)
    }{
      \varctx_1 \vdash \langle \{^{m}\, \overline{\term}\} : \typ \rangle^{\lft}_{\storectx} \dashv \varctx_3
    }
    \\[6mm]
    [\textsc{T-Declare}]
    \inference{
      \varctx_1 \vdash \langle \term : \typ \rangle^{\lft}_{\storectx} \dashv \varctx_2
      \\
      \dgr{x \notin \mathrm{dom}(\varctx_2)}
      \\
      \varctx_3 = \varctx_2[x \mapsto \slot{\typ}{\lft}]
    }{
      \varctx_1 \vdash \langle \kw{let}\ \kw{mut}\ x = \term : \epsilon \rangle^{\lft}_{\storectx} \dashv \varctx_3
    }
    &
    [\textsc{T-Assign}]
    \inference{
      \varctx_1 \vdash \langle \term : \typ_2 \rangle^{\lft}_{\storectx} \dashv \varctx_2
      \quad
      \dgr{\varctx_2} \vdash \lval : \slot{\ptyp_1}{m}
      \\
      \varctx_2 \vdash \ptyp_1 \shapecompat \typ_2
      \quad
      \varctx_2 \vdash \typ_2 \outlives m
      \\
      \varctx_3 = \writefn{0}(\varctx_2, \lval, \typ_2)
      \\
      \neg\writeProh(\varctx_3, \lval)
    }{
      \varctx_1 \vdash \langle \lval = \term : \epsilon \rangle^{\lft}_{\storectx} \dashv \varctx_3
    }
  \end{array}
  \end{equation*}
  \vspace{-3mm}
  \caption{\FR's typing rules for terms~\citep{pearce2021featherweight}. \dgr{\text{Grey boxes}} highlight corrections made by our mechanization.}
  \label{fig:typing}
\end{figure}

%% file: sections-new/method.tex
\newcommand{\uwidehat}[1]{%
  \mathpalette\douwidehat{#1}%
}
\makeatletter
\newcommand{\douwidehat}[2]{%
  \sbox0{$\m@th#1\widehat{\hphantom{#2}}$}%
  \sbox2{$\m@th#1x$}
  \sbox4{$\m@th#1#2$}
  \dimen0=\ht0
  \advance\dimen0 -.57\ht2
  \dimen2=\dp4
  \rlap{%
    \raisebox{\dimexpr\dimen0-\dimen2}{%
      \makebox[\wd4][c]{\scalebox{0.75}[-1]{\box0}}%
    }%
  }%
  {#2}%
}
\makeatother
\newcommand{\ppartial}[1]{\ensuremath{\uwidehat{#1}}}
\newcommand{\hole}{\ensuremath{\ppartial{\phantom{x}}}}
\newcommand{\pterm}{\ensuremath{\ppartial{t}}}
\newcommand{\ptermp}{\ensuremath{\ppartial{t'}}}
\newcommand{\plval}{\ensuremath{\ppartial{w}}}
\providecommand{\realizes}{\rightsquigarrow}
\providecommand{\lang}{\mathcal{L}}
\providecommand{\Compile}{\ensuremath{\mathcal{C}}}
\providecommand{\Pcompile}{\ensuremath{\mathcal{G}}}
\newcommand{\sealorfr}{\ensuremath{\sealor_{\text{\FR}}}}
\newcommand{\frlang}{\ensuremath{\lang_{\text{\FR}}}}
\newcommand{\Compilefr}{\ensuremath{\Compile_{\text{\FR}}}}
\newcommand{\Pcompilefr}{\ensuremath{\Pcompile_{\text{\FR}}}}

\input{figures/partial-realization-fr.tex}

\section{Instantiating Generative Compilation on \FR}
\label{sec:method}

In this section, we instantiate the generic framework of generative compilation (\cref{sec:framework}) by defining a concrete sealor for \FR.
We deliberately make this sealor syntax-guided, keeping both its implementation and its proof lightweight, thereby making the approach easier to adopt for real languages.
Despite this, we show that it is globally complete (\cref{sec:fr-preservation}) and sound on an important class of partial programs (\cref{sec:fr-soundness}).
The developments in this section are fully mechanized in Lean.

\subsection{Partial Syntax and Realization for \FR}

We first define \FR's \emph{partial syntax}, which lays the basis for the sealor definition.
It describes prefixes of syntactically valid \FR programs.
We then define \emph{realization}, which relates a partial-syntax string to a full-syntax string obtained by extending the partial string.
This relation is useful because completeness and soundness connect partial programs with their possible full programs.

\paragraph{Partial Syntax}
We define partial syntax relative to \FR's full syntax in \cref{fig:grammar}, focusing on source-level terms that may arise during autoregressive generation.
\cref{fig:partial-grammar} gives the resulting grammar of partial terms $\pterm$, which extends over a full term by allowing one of its parts to be only partially generated.
The first two cases, $\hole$ and $\term$, respectively describe the two endpoints: no part of the current term has yet been generated, and generation has produced a full term.
Each remaining production follows one term form from \cref{fig:grammar}, replacing the component currently being generated with its partial counterpart. 
For example, $\ppartial{\val}$ denotes a partial value, and $\kw{copy}\ \plval$ a copy term with a partial lval.
$\{^{\lft}\,\overline{\term};\,\pterm$ represents an unclosed block with a partial tail.
The lifetime label $\lft$ is an internal annotation, fixed once the opening brace is generated, rather than text generated by the model.
The production $\plval$ covers both moves and assignments, which share a prefix until an $=$ is generated.

Partial lvals $\plval$ and partial values $\ppartial{\val}$ complete the picture.
They are not listed in \cref{fig:partial-grammar} but are described in prose next.
A partial lval is either $\hole$, a full $\lval$, a partially generated variable name $\ppartial{x}$, or a dereference $*\plval$ whose operand remains partial.
A partial value $\ppartial{\val}$ is either a full value $\val$ or a partially generated integer constant; the unit value $\epsilon$ is always full.

For simplicity, we omit prefixes concerning only keywords, such as $\kw{let\ mut}$, $\kw{\&mut}$, or prefixes of either.
Such prefixes contain no partially generated term, lval, or value for the sealor to preserve, so we fold them into the case of $\hole$.

\paragraph{Realization of Partial Programs}
Realization models autoregressive generation: decoding may add text at the unfinished frontier, but cannot revise an earlier prefix.
It differs from sealing, which may discard or alter unfinished fragments.
This distinction will be made concrete once we introduce $\sealorfr$.
Our completeness result presented later in this section connects these two notions.

We write $\pterm \realizes \term$ when $\term$ realizes $\pterm$, and define this relation case by case in \cref{fig:realizes}.
In most cases, realization retains the term form already determined by the generated prefix and recurses on its unfinished subterm.
The two clauses for $\plval$ make its Move/Assignment ambiguity explicit: once $\plval$ realizes an lval $\lval$, the term may realize either the move $\lval$ or an assignment $\lval = \term$ with any right-hand side $\term$.
For lvals, the relation $\plval \realizes \lval$ has four cases: $\hole \realizes \lval$ for any $\lval$; $\lval \realizes \lval$; $\ppartial{x} \realizes x$ when $\ppartial{x}$ is a prefix of $x$; and $*\plval \realizes *\lval$ when $\plval \realizes \lval$.

\input{figures/sealor-fr.tex}
\subsection{$\sealorfr$: A Syntax-Guided Sealor for \FR}
\label{sec:fr-sealor}

\cref{fig:sealor} defines our syntax-guided sealor $\sealorfr$ by case analysis over partial syntax.
It maps each partial term form to a full \FR term.
As required by the generic definition, $\sealorfr$ is total.
Prefixes outside the partial syntax are handled separately by returning them unchanged, so that the underlying compiler may report a syntax error.

At a high level, $\sealorfr$ abstracts type checking through a syntax-guided transformation.
Rather than deciding exactly whether a partial term admits a well-typed completion, it constructs a term whose well-typedness is necessary for that of every possible completion, thereby preserving completeness.
Accordingly, sealing preserves generated structure when it exposes useful typing obligations, while abstracting away other obligations when they require information not yet generated or analysis beyond syntax.
Several cases therefore seal directly to $\eps$: $\kw{copy}\ \plval$, $\plval$, $\bshr\,[\kw{mut}]\ \plval$, and $\kw{let}\ \kw{mut}\ \ppartial{x}$.
Their unfinished components must be resolved before the corresponding typing rules can establish useful premises.
Here, $\eps$ provides a well-typed placeholder with no additional typing obligations.
Moreover, for $\kw{let}\ \kw{mut}\ x = \pterm$ and $\lval = \pterm$, $\sealorfr$ recurses only on $\pterm$.
It drops the enclosing declaration or assignment, including the declared variable $x$ or assigned lval $\lval$, and hence their associated typing obligations.

This abstraction is deliberately not maximally aggressive.
With completeness preserved, we seek to retain as much soundness as possible.
Fully generated terms are therefore unchanged, including both a standalone $\term$ and all already-full terms $\overline{\term}$ inside an unclosed block.
Furthermore, $\sealorfr$ preserves partial subterms that can themselves contain generated structure relevant to later typing, rather than discarding them outright.
Thus, it recurses on $\pterm$ in $\kw{box}\ \pterm$, $\{^{\lft}\,\overline{\term};\,\pterm$, $\kw{let}\ \kw{mut}\ x = \pterm$, and $\lval = \pterm$.
At the same time, we do not make $\sealorfr$ more sound at the cost of a lightweight design.
For instance, in $\kw{let}\ \kw{mut}\ x = \pterm$ or $\lval = \pterm$, one could use the type of $x$ or $\lval$ to further constrain $\pterm$.
We do not pursue this in $\sealorfr$, because it would require typing information beyond our syntax-guided design.
This would add implementation cost without substantial benefit, since $\sealorfr$ is already sound on an important class of programs, as we show in \cref{sec:fr-soundness}.

A few remaining cases complete the definition.
$\hole$ seals to $\eps$, since either no part of the current term has been generated or the generated prefix consists only of keywords.
An unfinished partial value $\ppartial{\val}$ also seals to $\eps$.
Although a partially generated integer already suggests type $\kw{int}$, we choose this simpler syntax-guided rule rather than recover that information, at a small cost to soundness.
The bare $\plval$ production likewise seals to $\eps$: before an $=$ is generated, it remains ambiguous between a move and the left-hand side of an assignment, and sealing it as either would prematurely commit to one interpretation.
Finally, for $\{^{\lft}\,\overline{\term};\,\pterm$, the sealor appends a trailing $\eps$ and the closing brace after sealing $\pterm$, ensuring that the result is a closed block.

\subsection{Global Completeness of $\sealorfr$}
\label{sec:fr-preservation}

$\sealorfr$'s design was motivated by global completeness, which we formally prove next.
We first show that whenever a partial term realizes a well-typed full term, sealing the partial term with $\sealorfr$ also yields a well-typed term.

\begin{theorem}[$\sealorfr$ Maintains Well-Typedness of Realizations]\label{thm:seal-preservation}
For all partial terms $\pterm$, typing environments $\varctx, \varctx'$, types $\typ$, lifetimes $\lft$, and store typings $\storectx$:
\[
\exists \term.\ \pterm \realizes \term\, \wedge\, \varctx \vdash \langle \term : \typ \rangle^{\lft}_{\storectx} \dashv \varctx' \implies \exists \typ^\star, \varctx^\star.\ \varctx \vdash \langle \sealorfr(\pterm) : \typ^\star \rangle^{\lft}_{\storectx} \dashv \varctx^\star.
\]
\end{theorem}

The proof follows \cref{fig:partial-grammar} row by row, in order, naming each case after that row's description.

\begin{proof}
Fix $\pterm$ and, using the hypothesis, a witness $\term$ with $\pterm \realizes \term$ and $\varctx \vdash \langle \term : \typ \rangle^{\lft}_{\storectx} \dashv \varctx'$. We proceed by induction on the derivation of $\pterm \realizes \term$.
\begin{itemize}[leftmargin=*]

\item \emph{Immediately sealed to $\eps$}: $\pterm = \hole$, $\pterm = \ppartial{\val}$, $\pterm = \kw{copy}\,\plval$, $\pterm = \plval$, $\pterm = \bshr\,[\kw{mut}]\,\plval$, or $\pterm = \kw{let}\ \kw{mut}\ \ppartial{x}$. In every case $\sealorfr(\pterm) = \eps$ unconditionally, well-typed by \textsc{T-Const}.

\item \emph{Full term}: $\pterm = \term$. Here $\sealorfr(\pterm) = \term$, and the conclusion follows immediately from the premise.

\item \emph{Partial heap allocation}: $\pterm = \kw{box}\,\ptermp$, so $\term = \kw{box}\,\term'$ with $\ptermp \realizes \term'$. By \textsc{T-Box}, $\term'$ is well-typed in $\varctx$; by the IH, so is $\sealorfr(\pterm) = \sealorfr(\ptermp)$.

\item \emph{Unclosed block with partial tail}: $\pterm = \{^{m}\, \term_1;\, \ldots;\, \term_{k-1};\, \pterm_k$, so $\term = \{^{m}\, \term_1;\, \ldots;\, \term_n \}$ for some $n \ge k$ with $\pterm_k \realizes \term_k$, and by definition $\sealorfr(\pterm) = \{^{m}\, \term_1;\, \ldots;\, \term_{k-1};\, \sealorfr(\pterm_k);\, \eps \}$. Inverting the typing of $\term$ via \textsc{T-Block} and \textsc{T-Seq} gives contexts $\varctx = \varctx_0, \varctx_1, \ldots, \varctx_n$ and types $\typ_1, \ldots, \typ_n$ with $\varctx_{i-1} \vdash \langle \term_i : \typ_i \rangle^{m}_{\storectx} \dashv \varctx_i$ for every $i$. By the IH, $\sealorfr(\pterm_k)$ is well-typed in $\varctx_{k-1}$; a final $\eps$ typed by \textsc{T-Const} closes the sequence by \textsc{T-Seq}, and \textsc{T-Block} then gives $\sealorfr(\pterm)$ well-typed in $\varctx$.

\item \emph{Declaration, partial init}: $\pterm = \kw{let}\ \kw{mut}\ x = \ptermp$, so $\term = \kw{let}\ \kw{mut}\ x = \term'$ with $\ptermp \realizes \term'$. By \textsc{T-Declare}, $\term'$ is well-typed in $\varctx$; by the IH, so is $\sealorfr(\pterm) = \sealorfr(\ptermp)$.

\item \emph{Assignment, partial RHS}: $\pterm = \lval = \ptermp$, so $\term = \lval = \term'$ with $\ptermp \realizes \term'$. By \textsc{T-Assign}, $\term'$ is well-typed in $\varctx$; by the IH, so is $\sealorfr(\pterm) = \sealorfr(\ptermp)$.

\end{itemize}

\end{proof}

Let $\frlang$ denote the set of well-typed top-level \FR terms.
As discussed in \cref{sec:lang}, \FR's top-level terms are typed under the root lifetime $\lft_{\mathrm{root}}$, empty typing environment, and empty store typing.
That is, $\frlang := \{\, \term \mid \exists \typ, \varctx.\ \emptyset \vdash \langle \term : \typ \rangle^{\lft_{\mathrm{root}}}_{\emptyset} \dashv \varctx \,\}$.
\cref{thm:seal-preservation} is stronger than the top-level result we ultimately need: it permits arbitrary initial typing environments and store typings, as well as lifetime.
This allows the proof to apply to intermediate terms within larger programs.
Specializing it to top-level typing yields completeness for top-level partial syntax.

\begin{corollary}[$\sealorfr$ is Globally Complete for Partial Syntax]
\label{cor:partial}
For every partial term $\pterm$,
\[
\exists \term.\ \pterm \realizes \term\, \wedge\, \term \in \frlang \implies \sealorfr(\pterm) \in \frlang.
\]
\end{corollary}

\begin{proof}
Fix a partial term $\pterm$ satisfying the premise, and choose a term $\term$ such that $\pterm \realizes \term$ and $\term \in \frlang$. By the definition of $\frlang$, there exist $\typ$ and $\varctx$ such that $\emptyset \vdash \langle \term : \typ \rangle^{\lft_{\mathrm{root}}}_{\emptyset} \dashv \varctx$.
Applying \cref{thm:seal-preservation} with initial typing environment and store typing $\emptyset$ yields some $\typ^\star$ and $\varctx^\star$ such that $\emptyset \vdash \langle \sealorfr(\pterm) : \typ^\star \rangle^{\lft_{\mathrm{root}}}_{\emptyset} \dashv \varctx^\star$.
Hence $\sealorfr(\pterm) \in \frlang$.
\end{proof}

The generic framework defines global completeness over arbitrary strings rather than partial-syntax terms, so we finally lift \cref{cor:partial} to that setting.

\begin{theorem}[$\sealorfr$ Is Globally Complete for Arbitrary Strings]
\label{thm:sealorfr-complete}
The sealor $\sealorfr$ is globally complete with respect to \FR, for arbitrary strings as input, in the sense of \cref{def:sealor-sound-complete}.
\end{theorem}

\begin{proof}
Let $s \in \alphabet^\ast$ satisfy $\exists w.\ s \concat w \in \frlang$. As $s \concat w$ is a well-typed \FR term, parsing $s$ yields a partial term $\pterm$, with keyword-only prefixes represented by $\hole$. And parsing $s \concat w$ yields a term $\term$ such that $\pterm \realizes \term$ and $\term \in \frlang$. By \cref{cor:partial}, $\sealorfr(\pterm) \in \frlang$. This is exactly the result of sealing $s$. Hence $\sealorfr$ accepts every string with a well-typed \FR continuation.
If $s$ does not yield a partial term, then it contains a syntax error and cannot satisfy $\exists w.\ s \concat w \in \frlang$. Here, $\sealorfr$ returns $s$ unchanged.
\end{proof}

Finally, based on \cref{thm:cons-prefix}, the global completeness of the corresponding generative compiler $\Pcompilefr(s) \coloneq \Compilefr(\sealorfr(s))$ follows from that of $\sealorfr$ and an exact compiler $\Compilefr$ for \FR.

\subsection{Selective Soundness of $\sealorfr$}
\label{sec:fr-soundness}

While $\sealorfr$ achieves global completeness, it also abstracts away information that soundness would need in order to remain lightweight and syntax-guided. For instance, $\sealorfr(\kw{copy}\ \ppartial{x}) = \eps$ is accepted unconditionally, even when $\ppartial{x}$ can only realize a moved-out variable.
We nevertheless identify an important class on which $\sealorfr$ is sound: \emph{statement boundaries}, where generation has finished one or more complete statements in a block and is about to begin the next.
This is a useful guarantee in practice because statement boundaries occur frequently during generation.

We now give formal characterizations of this property.
We call $\pterm$ a \emph{statement boundary} if $\pterm = \{^{\lft}\, \overline{\term};\, \hole$.
That is, after a series of full terms in a block, the model is exactly at the point immediately after a semicolon, before any token of what follows has streamed in. Let $\pclass_{\textsc{stmt}} \coloneq \{\, \{^{\lft}\, \overline{\term};\, \hole \mid \overline{\term} \text{ complete}, \ \lft \text{ a lifetime} \,\}$, i.e., the class of statement boundaries.

\begin{lemma}[Statement Boundaries Realize $\sealorfr$'s Output]\label{lem:stmt-self-witness}
For every $\pterm \in \pclass_{\textsc{stmt}}$, $\pterm \realizes \sealorfr(\pterm)$.
\end{lemma}
\begin{proof}
Write $\pterm = \{^{\lft}\, \overline{\term};\, \hole$, such that $\sealorfr(\pterm) = \{^{\lft}\, \overline{\term};\, \sealorfr(\hole);\, \eps \} = \{^{\lft}\, \overline{\term};\, \eps;\, \eps \}$. Instantiate the unclosed block case of realization $\realizes$ with $\pterm = \hole$ completing to $\term \coloneq \eps$ and $\overline{\term}_1 \coloneq \eps$. This gives exactly $\{^{\lft}\, \overline{\term};\, \hole \realizes \{^{\lft}\, \overline{\term};\, \eps;\, \eps \}$, i.e., $\pterm \realizes \sealorfr(\pterm)$.
\end{proof}

\begin{theorem}[$\sealorfr$ Reflects Well-Typedness onto Realizations at Statement Boundaries]\label{thm:selective-soundness}
For all $\pterm \in \pclass_{\textsc{stmt}}$, typing environments $\varctx, \varctx^\star$, types $\typ^\star$, lifetimes $\lft$, and store typings $\storectx$:
\[
\varctx \vdash \langle \sealorfr(\pterm) : \typ^\star \rangle^{\lft}_{\storectx} \dashv \varctx^\star \implies \exists \term, \typ, \varctx'.\ \pterm \realizes \term\, \wedge\, \varctx \vdash \langle \term : \typ \rangle^{\lft}_{\storectx} \dashv \varctx'.
\]
\end{theorem}
\begin{proof}
By \cref{lem:stmt-self-witness}, $\pterm \realizes \sealorfr(\pterm)$. By the hypothesis, $\sealorfr(\pterm)$ is a well-typed completion.
\end{proof}

As with \cref{thm:seal-preservation}, \cref{thm:selective-soundness} is stronger than the top-level result we need.
Specializing it yields selective soundness at top level.

\begin{corollary}[$\sealorfr$ is Sound for Partial Syntax at Statement Boundaries]\label{cor:stmt-partial}
For every $\pterm \in \pclass_{\textsc{stmt}}$, $\sealorfr(\pterm) \in \frlang \implies \exists \term.\ \pterm \realizes \term\, \wedge\, \term \in \frlang$.
\end{corollary}

\begin{proof}
Fix $\pterm \in \pclass_{\textsc{stmt}}$ with $\sealorfr(\pterm) \in \frlang$. By the definition of $\frlang$, there exist $\typ^\star$ and $\varctx^\star$ such that $\emptyset \vdash \langle \sealorfr(\pterm) : \typ^\star \rangle^{\lft_{\mathrm{root}}}_{\emptyset} \dashv \varctx^\star$.
Applying \cref{thm:selective-soundness} with initial typing environment, store typing $\emptyset$, and ambient lifetime $\lft_{\mathrm{root}}$ yields that $\pterm \realizes \term$ and $\emptyset \vdash \langle \term : \typ \rangle^{\lft_{\mathrm{root}}}_{\emptyset} \dashv \varctx'$ for some $\typ, \varctx'$.
The latter is equivalent to $\term \in \frlang$.
\end{proof}

We lift \cref{cor:stmt-partial} from partial-syntax terms to arbitrary strings.

\begin{theorem}[$\sealorfr$ Is Sound at Statement Boundaries, for Arbitrary Strings]\label{thm:selective-soundness-strings}
The sealor $\sealorfr$ is sound with respect to \FR on $\pclass_{\textsc{stmt}}$, for strings as input, in the sense of \cref{def:sealor-sound-complete}.
\end{theorem}

\begin{proof}
Let $s \in \pclass_{\textsc{stmt}}$ satisfy $\sealorfr(s) \in \frlang$. Since $\sealorfr(s)$ is a well-typed \FR term, parsing $s$ yields a partial term $\pterm \in \pclass_{\textsc{stmt}}$, with $\sealorfr(s) = \sealorfr(\pterm)$. By \cref{cor:stmt-partial}, $\exists \term.\ \pterm \realizes \term \wedge \term \in \frlang$. This $\term$ is realized by some string $s \concat w$. Hence $s$ is completable.
\end{proof}

Finally, based on \cref{thm:sealorfr-complete} and \cref{thm:selective-soundness-strings}, $\sealorfr$ is exact at statement boundaries. So is $\Pcompilefr$ according to \cref{thm:cons-prefix}.

%% file: figures/partial-realization-fr.tex
\begin{figure*}[t]
  \centering
  \begin{minipage}[t]{0.44\textwidth}
    \small
    \centering
    \begin{tabular}{@{}l@{\hspace{1mm}}l@{}}
$\pterm ::= $                                    & Partial term \\
\quad $\hole$                                    & \quad Nothing decoded \\
\quad $\term$                                    & \quad Full term \\
\quad $\ppartial{v}$                             & \quad Partial value \\
\quad $\kw{box}\ \pterm$                         & \quad Partial heap allocation \\
\quad $\kw{copy}\ \plval$                        & \quad Partial copy \\
\quad $\plval$                                   & \quad Partial move, or assign LHS \\
\quad $\bshr\,[\kw{mut}]\ \plval$                & \quad Partial (mutable) borrow \\
\quad $\{^{\lft}\, \overline{\term};\, \pterm$   & \quad Unclosed block, partial tail \\
\quad $\kw{let}\ \kw{mut}\ \ppartial{x}$         & \quad Declaration, partial name \\
\quad $\kw{let}\ \kw{mut}\ x = \pterm$           & \quad Declaration, partial init \\
\quad $\lval = \pterm$                           & \quad Assignment, partial RHS \\
    \end{tabular}
    \vspace{-2.5mm}
    \caption{Our partial syntax for \FR{} terms.}
    \label{fig:partial-grammar}
  \end{minipage}
  \hfill
  \begin{minipage}[t]{0.53\textwidth}
    \small
    \centering
    \begin{tabular}{@{}r@{\ }c@{\ }l@{\hspace{2mm}}l@{}}
$\hole$                                    & $\realizes$ & $\term$                                                             & any $\term$ \\
$\ppartial{\val}$                          & $\realizes$ & $\val$                                                              & if $\ppartial{\val}$ is a prefix of $\val$ \\
$\term$                                    & $\realizes$ & $\term$                                                             & unconditionally \\
$\kw{box}\ \pterm$                         & $\realizes$ & $\kw{box}\ \term$                                                   & if $\pterm \realizes \term$ \\
$\kw{copy}\ \plval$                        & $\realizes$ & $\kw{copy}\ \lval$                                                  & if $\plval \realizes \lval$ \\
$\plval$                                   & $\realizes$ & $\lval$                                                             & if $\plval \realizes \lval$ \\
$\plval$                                   & $\realizes$ & $\lval = \term$                                                     & if $\plval \realizes \lval$, any $\term$ \\
$\bshr\,[\kw{mut}]\ \plval$                & $\realizes$ & $\bshr\,[\kw{mut}]\ \lval$                                          & if $\plval \realizes \lval$ \\
$\{^{\lft}\, \overline{\term}_0;\, \pterm$ & $\realizes$ & $\{^{\lft}\, \overline{\term}_0;\, \term;\, \overline{\term}_1 \}$  & if $\pterm \realizes \term$, any $\overline{\term}_1$ \\
$\kw{let}\ \kw{mut}\ \ppartial{x}$         & $\realizes$ & $\kw{let}\ \kw{mut}\ x = \term$                                     & if $\ppartial{x}$ is a prefix of $x$, any $\term$ \\
$\kw{let}\ \kw{mut}\ x = \pterm$           & $\realizes$ & $\kw{let}\ \kw{mut}\ x = \term$                                     & if $\pterm \realizes \term$ \\
$\lval = \pterm$                           & $\realizes$ & $\lval = \term$                                                     & if $\pterm \realizes \term$ \\
    \end{tabular}
    \vspace{-3mm}
    \caption{Our realization relation for \FR{} terms ($\pterm \realizes \term$).}
    \label{fig:realizes}
  \end{minipage}
\end{figure*}

%% file: figures/sealor-fr.tex
\begin{wrapfigure}{r}{0.39\textwidth}
  \vspace{-14mm}
  \small
  \centering
  \begin{tabular}{@{}l@{\hspace{1mm}}l@{}}
$\sealorfr(\pterm) = $                                                        & \\
\quad $\eps$                                                                  & \quad if $\pterm = \hole$ \\
\quad $\term$                                                                 & \quad if $\pterm = \term$ \\
\quad $\eps$                                                                  & \quad if $\pterm = \ppartial{v}$ \\
\quad $\sealorfr(\ptermp)$                                                    & \quad if $\pterm = \kw{box}\ \ptermp$ \\
\quad $\eps$                                                                  & \quad if $\pterm = \kw{copy}\ \plval$ \\
\quad $\eps$                                                                  & \quad if $\pterm = \plval$ \\
\quad $\eps$                                                                  & \quad if $\pterm = \bshr\,[\kw{mut}]\ \plval$ \\
\quad $\{^{\lft}\, \overline{\term};\, \sealorfr(\ptermp);\, \eps \}$         & \quad if $\pterm = \{^{\lft}\, \overline{\term};\, \ptermp$ \\
\quad $\eps$                                                                  & \quad if $\pterm = \kw{let}\ \kw{mut}\ \ppartial{x}$ \\
\quad $\sealorfr(\ptermp)$                                                    & \quad if $\pterm = \kw{let}\ \kw{mut}\ x = \ptermp$ \\
\quad $\sealorfr(\ptermp)$                                                    & \quad if $\pterm = \lval = \ptermp$ \\
  \end{tabular}
  \vspace{-3mm}
  \caption{Our syntax-guided sealor $\sealorfr$.}
  \label{fig:sealor}
  \vspace{-3mm}
\end{wrapfigure}

%% file: sections-new/rust.tex
\providecommand{\todoterm}{\ensuremath{\kw{todo}}}
\providecommand{\panicterm}{\ensuremath{\kw{hole}_\kw{div}()}}
\providecommand{\missingterm}{\ensuremath{\kw{hole}_\kw{val}()}}
\providecommand{\kw}[1]{\ensuremath{\mathtt{#1}}}
\providecommand{\diverges}{\mathbin{\Uparrow}}
\providecommand{\gpar}[1]{\ensuremath{\textrm{<}#1\textrm{>}}}
\providecommand{\prog}{\ensuremath{p}}

\providecommand{\blck}{\ensuremath{b}}
\providecommand{\pblck}{\ensuremath{\uwidehat{b}}}

\newcommand{\sealorrs}{\ensuremath{\sealor_{\text{RS\xspace}}}}
\newcommand{\sealorrse}{\ensuremath{\sealorrs^e}}
\newcommand{\sealorrss}{\ensuremath{\sealorrs^s}}

\section{From \FR to Real Rust}
\label{sec:rust}

We now carry the methodology of \cref{sec:method} from \FR to real Rust by presenting our Rust sealor, $\sealorrs$.

\paragraph{Inherited from \FR}
We design $\sealorrs$ following the same principles as \FR's sealor $\sealorfr$ (\cref{sec:fr-sealor}).
It remains lightweight and syntax-guided in most cases, discards typing obligations when needed for completeness, and adds targeted analysis where doing so improves soundness.
It preserves already-generated structure verbatim and recurses on the still-partial frontier.
As with $\sealorfr$, the goal is global completeness together with selective soundness.

For each Rust feature, we describe its \emph{syntactic shape}, the \emph{partial forms} the sealor may encounter, and the \emph{sealor rules} used to close those partial forms.
For readability, we reuse \FR's notation where applicable and introduce new notation only for constructs specific to Rust.

\paragraph{Not Inherited from \FR}
Two aspects of this section depart from \cref{sec:method}.
First, Rust is much larger than \FR.
The Rust sealor must therefore handle constructs with no \FR counterpart, as well as constructs that resemble \FR but differ in syntactic structure or typing behavior.
We discuss a representative selection of such features across this section.
Second, we give informal, per-feature arguments for typing behavior, completeness, and soundness, rather than the formal, mechanized treatment we gave for the whole \FR.
This is because our goal for this section is to directly capture the behaviors of \rustc{}, the Rust compiler.
Using an existing Rust formalization as an intermediate target would make correctness depend on that formalization's faithfulness to \rustc{}, rather than on \rustc{} directly, which is especially undesirable as Rust continues to evolve.
The global completeness argument follows the same structure as \cref{sec:fr-preservation}: as long as each sealor case is complete with respect to the corresponding \rustc{} behavior, the entire sealor and the generative compiler are complete.

\subsection{General Design Choices for Rust}
\label{sec:rust-distinctions}

We describe design choices recurring across Rust's features, before per-feature treatment in \cref{sec:extension}.

\paragraph{Sealors $\sealorrss$ and $\sealorrse$}
Since Rust distinguishes statements and expressions,
we introduce two sealors:
a statement sealor $\sealorrss$ that seals a partial statement (possibly an expression statement) into a list of statements,
and an expression sealor $\sealorrse$ that seals a partial expression into an expression with an unconstrained type.
When an expression is desired, and its type can be inferred from context, the expression sealor is used.
In all other cases, we use the statement sealor.

\paragraph{Placeholders $\panicterm$ and $\missingterm$}
\FR's sealor closes every abstracted case with the single unit placeholder $\eps$, since \FR's control-flow-free terms leave no other typing obligation to discharge.
\rustc{}'s typing rules for branching and looping constructs additionally require reasoning about multiple control-flow paths, e.g., that both branches of an $\kw{if}$ agree, or that a loop body's effects are consistent across iterations.
To seal such constructs conveniently, we introduce two placeholders with deliberately different typing behavior.

The first placeholder, $\panicterm$, is simply $\kw{panic!}()$ in our implementation.
Statically, $\kw{panic!}()$ has type $!$, the never type, whose expressions can be coerced to any expected type.
It diverges: executing it never returns normally.
Because the continuation after a divergent expression is unreachable, \rustc{} does not require it to satisfy the usual borrow-checking obligations.
This helps preserve completeness; otherwise \rustc{} might report borrow errors caused only by an intermediate partial state, such as a variable being moved out before later generated code would reinitialize it.
This makes $\panicterm$ useful for sealing control flow branches that should not produce a value.

The second, $\missingterm$, instead behaves as if it produces a well-typed value without diverging.
We implement it as a call to a generic helper function:
$\kw{const}\ \kw{fn}\ \missingterm{}\gpar{T}() \rightarrow T \ \{\ \kw{panic!}()\ \}$.
The function body diverges at run time, but this divergence is not exposed at the call site.
The call expression $\missingterm$ has the declared return type $T$, not the body type $!$.
The type parameter $T$ is instead resolved by type inference from the surrounding context.
Thus, $\missingterm$ can instantiate to any expected type in the surrounding code while keeping the borrow checker enabled.

A more subtle reason not to use $\panicterm$ as the expression placeholder is Rust's fallback for the never type.
When type inference is underconstrained, $!$ may fall back to the unit type $()$, causing spurious type or trait mismatch errors.
The type parameter $T$ in $\missingterm$ does not undergo this fallback, making it a safe choice when sealed code must supply an expression value.

\paragraph{Suppressing Future-Dependent Errors During Generation}
\label{sec:expected-errors}
Sealing Rust sometimes requires treating \rustc{} errors more carefully than a simple accept-or-reject verdict.
This is not a peculiarity of our sealor, but a consequence of Rust features whose well-typedness can depend on code generated later.
For example, a partial trait implementation may not yet contain all required methods, a call may refer to a function item declared later in the file, and an expression may have an ambiguous type that is resolved only by a subsequent use.
Such errors are required to be suppressed for completeness.
Otherwise the sealed program could be rejected even though the missing suffix would make the final program well-typed.
To limit the impact on soundness, we suppress errors only in the scope where necessary.
For instance, ``missing trait items'' is suppressed only if it occurs inside the sealing of partial trait implementations and ``type annotations needed'' errors only inside the sealing of partial function bodies.
When we arrive at a fully generated program after the $\EOS$ token, we clean out all suppressed errors, such that \rustc{} could catch all errors.

\paragraph{Projecting Error Messages on Partial Programs}
Since \rustc{} checks the sealed program $\sealorrs(c)$ rather than the original partial program $c$, its diagnostics refer to $\sealorrs(c)$.
Feeding these diagnostics directly to the model can be confusing, because $\sealorrs(c)$ may differ from $c$, as in \cref{fig:bg-partial,fig:bg-sealed}.
We therefore project diagnostics from $\sealorrs(c)$ back to $c$.
During sealing, we maintain a positional map of code copied verbatim from $c$ into $\sealorrs(c)$.
We then re-render \rustc{}'s diagnostics using $c$ as the reference code, mapping the locations of diagnostic highlights from $\sealorrs(c)$ back into $c$.
Inserted code such as $\missingterm{}$ has no source in $c$; diagnostics arising only from such code are suppressed by the module described above.

\subsection{$\sealorrs$: Our Sealor for Rust}
\label{sec:extension}

We now introduce the Rust sealor rules for expression and statement forms.

\paragraph{Blocks}
A block in Rust, denoted $\blck$, consists of a sequence of statements $\overline{s}$, optionally terminated by an expression $e$.
That is, $\blck ::= \{\, \overline{s};\, e^? \}$.
A block is itself an expression, so it may appear wherever an expression is expected.
A block creates a lexical scope for bindings: variables defined inside are not visible outside, and owned values are dropped at the end unless moved earlier.
The type of a block is the type of its tail expression, or the unit type if there isn't one.
Independently of that, a block diverges if any of its statements or tail expression diverges.
A block may end with a partial expression or statement. Both are sealed identically:
\begin{center}
\vspace{1mm}
\begin{tabular}{@{}l@{\hspace{1mm}}l@{\hspace{4mm}}|@{\hspace{4mm}}l@{\hspace{1mm}}l@{}}
$\sealorrss(\pblck) = $                                              &                                        & $\sealorrse(\pblck) = $                                              & \\
\quad $\{\, \overline{s};\, \sealorrss(\ppartial{e});\, \panicterm \,\}$  & \quad if $\pblck = \{\, \overline{s}\ \ppartial{e}$  & \quad $\{\, \overline{s};\, \sealorrss(\ppartial{e});\, \missingterm \,\}$  & \quad if $\pblck = \{\, \overline{s}\ \ppartial{e}$ \\
\quad $\{\, \overline{s};\, \sealorrss(\ppartial{s});\, \panicterm \,\}$  & \quad if $\pblck = \{\, \overline{s}\ \ppartial{s}$  & \quad $\{\, \overline{s};\, \sealorrss(\ppartial{s});\, \missingterm \,\}$  & \quad if $\pblck = \{\, \overline{s}\ \ppartial{s}$
\end{tabular}
\vspace{1mm}
\end{center}
The block sealors preserve the full statements $\overline{s}$ and seal the partial component with $\sealorrss$.
The statement sealor then appends $\panicterm$, making the sealed block diverge to stop the control-flow path (more on this later in this section when we discuss conditionals and loops).
The expression sealor instead appends $\missingterm$, so the block produces a value of the type expected by its context.

Completeness follows from two facts.
First, preserving the completed statements $\overline{s}$ and recursively sealing the partial component ($\sealorrss(\ppartial{s})$) mirrors \FR's completeness argument for unclosed blocks (\cref{sec:fr-preservation}).
That is, each piece is checked in exactly the environment it would occupy in any well-typed completion.
Second, the trailing placeholder never introduces a new failure: $\panicterm$ and $\missingterm$ type check against any expected type, so appending them cannot make an otherwise completable block fail to type check.
For soundness, we preserve all of \rustc{}'s checking on $\overline{s}$ and on the recursively sealed partial form, so any uncompletable prefix that fails these checks is rejected.

\paragraph{Fallback Sealors}
For every feature discussed next, we present a bespoke statement sealor $\sealorrss$.
We do not need a separate expression sealor $\sealorrse$, since we automatically derive one from $\sealorrss$ via $\sealorrse(\ppartial{e}) \coloneq \{\, \sealorrss(\ppartial{e});\, \missingterm\, \}$.
This turns any sealing of $\ppartial{e}$ from a list of statements into an expression that produces a value of the expected type.
Both completeness and soundness of this derivation are inherited directly from $\sealorrss$: appending $\missingterm$ neither introduces nor discards any typing obligation, since $\missingterm$ type checks against any expected type.

When a construct has no specific rule for $\sealorrss(\ppartial{s})$, we fall back to $\sealorrss(\ppartial{s}) \coloneq ()$.
This discards the partial statement, preserving completeness by removing its typing obligations.
Soundness may be lost temporarily, but is recovered once the statement is fully generated and copied verbatim.

\paragraph{Conditionals}
\label{sec:ext-cond}
Conditionals in Rust are expressions: the condition is an expression, and the branches are blocks; chaining multiple conditions with \kw{else\ if} is allowed.
That is, $e_\kw{if} ::= \kw{if}\ e\ \blck \sep{} \kw{if}\ e\ \blck\ \kw{else}\ \blck \sep{} \kw{if}\ e\ \blck\ \kw{else}\ e_\kw{if}$.
\rustc{} type checks the condition against $\kw{bool}$, then type checks each branch and merges the two states after them.
When an \kw{else\ if} chain is nested, the two-way merge applies recursively from the innermost pair outward.
A branch that diverges contributes nothing to the merge: its result type and borrow state at the join point are unobservable, so the result type and the borrow state after the $\kw{if}$ are entirely determined by the live branch.

A partial component may appear inside the condition, inside the then-branch, or inside the else-branch.
That is, $\ppartial{e_\kw{if}} ::= \kw{if}\ \ppartial{e} \sep{} \kw{if}\ e\ \pblck \sep{} \kw{if}\ e\ \blck\ \kw{else}\ \pblck \sep{} \kw{if}\ e\ \blck\ \kw{else}\ \ppartial{e_\kw{if}}$.
We close the missing branch with $\panicterm$, so that it is diverging and the whole $\kw{if}$ is well-typed against the live side alone.
The statement sealor $\sealorrss$ for $\ppartial{e_\kw{if}}$ is defined below:
\begin{center}
\vspace{1mm}
\begin{tabular}{@{}r@{\hspace{1mm}}l@{\hspace{1mm}}l@{}}
$\sealorrss(\ppartial{e_\kw{if}}) = $ & $\sealorrss(\ppartial{e})$                                                            & \quad if $\ppartial{e_\kw{if}} = \kw{if}\ \ppartial{e}$ \\
                                      & $\kw{if}\ e_0\ \sealorrss(\pblck)\ \kw{else}\ \{\, \panicterm \,\}$                   & \quad if $\ppartial{e_\kw{if}} = \kw{if}\ e_0\ \pblck$ \\
                                      & $\kw{if}\ e_0\ \blck_1\ \kw{else}\ \sealorrss(\pblck)$                                & \quad if $\ppartial{e_\kw{if}} = \kw{if}\ e_0\ \blck_1\ \kw{else}\ \pblck$ \\
                                      & $\kw{if}\ e_0\ \blck_1\ \kw{else}\ \{\, \sealorrss(\ppartial{e});\, \panicterm \,\}$  & \quad if $\ppartial{e_\kw{if}} = \kw{if}\ e_0\ \blck_1\ \kw{else\ if}\ \ppartial{e}$ \\
                                      & $\kw{if}\ e_0\ \blck_1\ \kw{else}\ \sealorrss(\kw{if}\ e_1 \ldots)$                   & \quad if $\ppartial{e_\kw{if}} = \kw{if}\ e_0\ \blck_1\ \kw{else\ if}\ e_1 \ldots$
\end{tabular}
\vspace{1mm}
\end{center}
A partial $\kw{else\ if}$ condition is sealed as an $\kw{else}$ body, closed with $\panicterm$ just like the missing branch itself.
Once the condition becomes full, we instead recurse via $\sealorrss$ on the remaining $\kw{if}\ e_1 \ldots$ as a fresh partial $\kw{if}$.

Our sealor is complete for conditionals.
Every partial component we recurse on inherits completeness from the inductive hypothesis on branch blocks and boolean conditions.
Although we do not separately discuss boolean conditions in detail, the completeness of sealing them should be easy to achieve.
The only other addition, $\panicterm$, never introduces new obligations: since it diverges, \rustc{} excludes it from the branch-type merge, so it cannot constrain what the live side must satisfy.
For soundness, we preserve all of Rust's checking up to the partial form, so any uncompletable prefix that fails one of these checks is still rejected.

\paragraph{While Loops}
A Rust \kw{while} loop consists of a condition and a body block.
That is, $e_\kw{while} ::= \kw{while}\ e\ \blck$.
A well-typed \kw{while} loop requires that the loop body be well-typed \emph{as if it were executed any number of times}.
Informally, the borrow and initialization state after one iteration must be compatible with the state at the start of the next.
Moving out of a variable in a loop body without reinitialization is rejected, as the second iteration would read a moved-out slot.
Rust waives this obligation when the body diverges, because control never reaches a second iteration.

Either the condition or the body may be partial, i.e., $\ppartial{e_\kw{while}} ::= \kw{while}\ \ppartial{e} \sep{} \kw{while}\ e\ \pblck$.
We force a partial body to diverge by appending $\panicterm$ via $\sealorrss(\pblck)$, eliminating the ``carries over between iterations'' obligation while still checking everything on the already-generated first iteration.
\begin{center}
\vspace{1mm}
\begin{tabular}{@{}r@{\hspace{1mm}}l@{\hspace{1mm}}l@{}}
$\sealorrss(\ppartial{e_\kw{while}}) = $ & $\sealorrss(\ppartial{e})$            & \quad if $\ppartial{e_\kw{while}} = \kw{while}\ \ppartial{e}$ \\
                                         & $\kw{while}\ e\ \sealorrss(\pblck)$   & \quad if $\ppartial{e_\kw{while}} = \kw{while}\ e\ \pblck$
\end{tabular}
\vspace{1mm}
\end{center}
The other loop forms in Rust (\kw{for}, unbounded \kw{loop}) are handled analogously.
For completeness, the condition case is similar to $\kw{if}$.
For the body case, the sealed body diverges by construction, so \rustc{} does not enforce the second-iteration consistency check, and the first-iteration prefix (the condition and the completed body prefix $\overline{s}$) is checked against exactly the same environment it would occupy in any well-typed completion.
For soundness, we do not lose anything more than what is lost in the partial term in the loop body, as the realization of $\pblck$ can also diverge.

\paragraph{Function Calls}
\label{sec:ext-call}
A function call consists of a callee and zero or more arguments.
That is, $e_\kw{call} ::= e(\overline{e})$.
At a call site $e_0(e_1, \ldots, e_n)$, \rustc{} resolves $e_0$ to a signature, then type checks the arguments left-to-right, threading each argument's move/borrow effects into the environment used to check the next.
The number of arguments must equal the declared arity.

A partial call has zero or more completed arguments followed by one partial argument.
That is, $\ppartial{e_\kw{call}} ::= e(\overline{e}, \ppartial{e}$.
If the callee is still partial, sealing reduces to handling whatever partial expression the callee is via its own rule, so we do not specifically discuss this case here.
Our sealor first queries \rustc{} as an \emph{arity oracle}: we invoke it on a sealing where $e_0$ is called with no arguments and inspect its type.
Once the arity $n$ is known, we pad with $\missingterm$ for the missing arguments: via $\sealorrse$ for the partial argument, and directly for the ones not at all present in $\ppartial{e_\kw{call}}$.
\begin{center}
\vspace{1mm}
\begin{tabular}{@{}l@{\hspace{1mm}}l@{}}
$\sealorrss(\ppartial{e_\kw{call}}) = e_0(e_1, \ldots, e_k, \sealorrse(\ppartial{e}), \underbrace{\missingterm, \ldots, \missingterm}_{n - k - 1 \text{ copies}})$ & \quad if $\ppartial{e_\kw{call}} = e_0(e_1, \ldots, e_k, \ppartial{e}$ \\
\end{tabular}
\vspace{1mm}
\end{center}

For completeness, our sealor places each completed argument $e_i$ in exactly the same environment it would occupy in any well-typed completion (arguments before it are the same; then the partial argument is sealed, and arguments after it are $\missingterm$, which is generic and produces no effect on the environment).
This covers the case $k < n$, i.e., fewer arguments have been generated so far than the callee accepts.
If instead $k \geq n$, the prefix already has more arguments than the callee's arity, our sealor still copies all of them rather than truncating, so \rustc{} rejects the sealed call with a genuine arity mismatch.
This is not a false rejection, since no well-typed completion can exist once too many arguments have already been committed.
For soundness, we preserve \rustc{}'s checking of the callee $e_0$ and the completed arguments $e_1, \ldots, e_k$ exactly as it would occur in the real program, so any type or borrow error among them is still caught.

\paragraph{Field Accesses}
\label{sec:ext-field}
A field access has the form $e.f$, where $e$ is an expression of record (or reference-to-record) type and $f$ is an identifier naming one of its fields.
That is, $e_\kw{field} ::= e.f$.
Field accesses are lvals, so they participate in move, copy, borrow, and assignment on the same footing as variables.

The partial syntax rule is $\ppartial{e_\kw{field}} ::= e.\ppartial{f}$.
In other words, the receiver is completed but the field has been cut off mid-token.
For $e.\ppartial{f}$, we only copy the receiver.
We take a reference to avoid moving out of $e$.
This is necessary because the completion may be a method invocation (either on $e$ itself or on a field) that takes $\kw{self}$ by reference.
That is, $\sealorrss(e.\uwidehat{f}) = \&e$.

For completeness, the receiver of any partial field access is a strict sub-prefix of every completion, so the receiver is well-typed if the completion is well-typed.
In terms of borrow checking, we need to suppress \texttt{E0382} (borrow of moved value) because fields of the receiver (other than the one accessed) could be uninitialized.
Suppressing \texttt{E0382} means we do not check if the field (or the whole record) is initialized until later.
For soundness, we preserve the type and borrow-checking of the receiver in all cases.
We do not attempt to check whether truncated identifiers can be extended to a real field of the receiver's type, as this goes beyond our syntax-guided design.
Instead, we wait until the field becomes full.

\paragraph{Other Expressions and Top-Level Items}
We briefly summarize the remaining Rust constructs. Aggregate literals, including tuples $(\overline{e})$, arrays $[\overline{e}]$, and record expressions $S\{\overline{f:e}\}$, are sealed by recursing on the trailing partial element with $\sealorrse$ and then closing the aggregate. For record expressions, we suppress \texttt{E0063} (missing field initializers). This preserves completeness because already generated elements are checked in order, while missing fields may still be supplied later; soundness is preserved for the fully generated elements.

Closure expressions are sealed by sealing the body as an expression, using $\missingterm$ to close value-producing positions. Closures with partial signatures are skipped until the signature is complete, which delays checking of parameter types. Other expression forms, including unary expressions, \kw{return}, condition-local \kw{let}, parenthesized expressions, \kw{try}, ranges, indexing, and casts, are handled uniformly: the sealor preserves fully generated components and recurses into the partial one. Constructs without a bespoke expression rule use the fallback sealor described above.

Top-level items are handled similarly. Function bodies are sealed like closure bodies. Item-containing constructs, such as module definitions and \texttt{impl} blocks, preserve the enclosing definition and all fully generated sub-items, then recursively seal the partial item at the frontier. For constructs without bespoke rules, including type definitions, type aliases, constants, statics, extern blocks, global \texttt{asm!}, and macros, fully generated instances are preserved verbatim, while partial instances are discarded until enough syntax has been generated for a more specific rule to apply.

\paragraph{Example Revisited}
Three rules seal the partial program in \cref{fig:bg-partial} into the full program in \cref{fig:bg-sealed}.
First, the function body is an unclosed block ending in the partial statement $\kw{let}\ \kw{prio} = \kw{tag}.\ppartial{f}$, so the block rule appends $\missingterm$: $\sealorrse(\pblck) = \{\, \overline{s}; \sealorrss(\ppartial{s});\, \missingterm \,\}$.
Second, like the \FR sealor (\cref{sec:fr-sealor}), the $\kw{let}$ rule drops the left-hand side and recurses on the partial right-hand side: $\sealorrss(\kw{let}\ \kw{prio} = \kw{tag}.\ppartial{f}) = \sealorrss(\kw{tag}.\ppartial{f})$.
Finally, the field-access rule closes the partial access as $\&\kw{tag}$.

%% file: sections-new/eval.tex
\section{Experimental Evaluation}
\label{sec:eval}

In this section, we experimentally evaluate our generative compilation approach.
We focus on two core research questions: (R1) Can our generative compilation effectively support code generation? and (R2) What are observable effects of generative compilation?

\subsection{Experimental Setup}

\paragraph{Models and Agent Harness}
We evaluate seven recent coding models.
The three frontier black-box models, \mopus{}~\citep{AnthropicClaudeOpus48SystemCard}, \mgptcodex{}~\citep{OpenAI2026Gpt53CodexSystemCard}, and \mgemini{}~\citep{GoogleDeepMind2026Gemini35FlashModelCard}, come from major model providers and are explicitly advertised for coding.
We also include four strong open-weight models, Kimi K2.7 Code \citep{kimiK27Code}, \mglm{} \citep{glm}, and two models from the \mqwen{} family, a large 397B-parameter model and a smaller 9B-parameter model~\citep{qwen35blog}.
This selection spans a broad range of model categories and capabilities.
We evaluate each model twice on every task, sampling with temperature $0.6$ whenever the provider permits configuring it.

We use an agent harness with fixed control flow, similar in style to Agentless~\citep{10.1145/3715754}, to isolate the effect of feedback and enable a clean comparison.
The model generates files or structured edits, while the harness deterministically writes them to disk, runs fixed compilation commands, and feeds any diagnostics back into the next prompt.

\paragraph{Tasks and Benchmarks}
We evaluate two repository-level tasks for which Rust compiler errors occur frequently.
Both tasks require filling a Rust skeleton that defines struct and function signatures.
We report two end-to-end metrics: \emph{compiler error rate} and \emph{functional correctness}.
Compiler error rate is the fraction of generated outputs for which the Rust compiler reports at least one error diagnostic.
Functional correctness is the fraction of generations that pass the task's unit tests.
To better understand the results, we also measure runtime, the number of diagnostics in each compiler-feedback report, whether errors are detected during generation, and how far detected errors are from their original locations.
We list prompts for both tasks in \cref{app:prompts}.

The first task, denoted \dtranslation{}, is based on CRUST-Bench~\citep{khatry2025crustbench}, a benchmark for C-to-Rust translation.
\dtranslation{} is based on a subset of 20 complex instances from CRUST-Bench.
The model is given the C implementation of a library and is tasked to translate it into Rust.
We ask the model to translate each C file at a time, such that it uses the specified interfaces.
C to Rust translation is an important topic, as Rust is considered a more memory-safe and similarly performant alternative to C.
At the same time, large translations are challenging: the model must translate C data structures consistently and adapt them to Rust interfaces across files.

The second task, denoted \dupgrade{}, evaluates whether models can adapt to changed Rust library APIs.
Each instance asks the model to implement a small command-line utility that requires using a recently updated API function.
This setting is relevant because LLMs always have fixed training knowledge cutoffs and may therefore generate code against deprecated APIs.
The benchmark tests whether the model can use compiler feedback about the changed API to repair its generation.
We provide details about benchmark construction in \cref{app:exp-details}.

\paragraph{Comparison}
We compare three approaches.
The first, \vanilla{}, samples directly from the LLM without compiler feedback.
The second, denoted \rfinals{} for short, adds standard post-generation compiler feedback: after each generated solution, the compiler is run, and any diagnostics are fed back to the model.
If compilation fails, the model may either edit its previous output or regenerate a new solution from scratch.
This process may repeat up to $n$ times.
The third approach, denoted \rimms{}, further incorporates on-the-fly feedback enabled by generative compilation (\cref{fig:gc-threads}).
This tests whether generative compilation can improve over post-generation feedback by detecting errors early.
During generation, generative compilation may detect compiler errors on partial outputs and restart generation, up to $k$ times.
If the model has not produced a valid output after these $k$ restarts, we allow it to finish generating a full output and then apply up to $n-k$ iterations of post-generation compiler feedback.
In our experiments, we set $n=15$ in \dupgrade{} and $n=20$ on \dtranslation{}, and $k=10$ in both settings.
We ablate over the budget choices in \cref{app:budget-ablation}.

Note that we cannot experimentally compare against constrained decoding because, to our knowledge, no constrained decoding implementation exists for Rust.
Although our approach can check partial programs, it is not suitable for constrained decoding.
This is because our approach guarantees global completeness, whereas constrained decoding requires global soundness, as discussed at the end of \cref{sec:defs}.
Building a constrained decoder from scratch for real Rust would be highly complex and costly, as discussed in \cref{sec:motivation}.
This contrast highlights a key benefit of generative compilation: it is lightweight to implement and largely reuses existing compiler infrastructure.



\paragraph{Implementation Details} Our sealor is implemented in about 5k lines of Rust, which is orders of magnitude less than the $\sim$600k lines in Rust's compiler frontend \citep{rustc}.
It parses partial Rust code using \textsf{rust-analyzer} \citep{rust-analyzer}, synthesizes sealed programs using the rules outlined in \cref{sec:rust}, and queries \rustc{} for compiler feedback. We apply minimal patches to rust-analyzer to improve its handling of partial ASTs.
The LLM-facing layer comprises about 3k lines of Python: it streams candidate tokens, invokes the sealor and the Rust compiler through bindings, and implements the logic in \cref{fig:gc-threads}.
For robustness, the implementation is covered by 329 Rust and Python tests spanning the sealor and inference behavior. Finally, we include wrappers for different LLM inference APIs, evaluated datasets, and inference and evaluation orchestration comprising another 8k lines of Python code.

\input{figures/table_downstream.tex}

\subsection{Benefits for End-to-End Code Generation}
\label{sec:eval:downstream}

We integrate our method into the code generation pipeline and measure its benefits, answering R1.

\paragraph{Reduced Compiler Errors and Improved Functional Correctness}
As shown in \cref{tab:downstream}, \vanilla{} produces non-compilable outputs in up to $90.0\%$ of cases.
Post-generation compiler feedback (\rfinals{}) reduces the average compiler error rate from $65.9\%$ to $20.7\%$, highlighting the strong repair capabilities of state-of-the-art LLMs.
Even so, adding on-the-fly feedback with \rimms{} reduces the error rate further to $13.1\%$.
Across all 14 model-dataset configurations, \rimms{} achieves a better compiler error rate than \rfinals{} and one tie;
in 9 of the configurations, the improvement over \rfinals{} is statistically significant.
Notably, \rimms{} eliminates compiler errors entirely for \mopusshort{} on \dupgrade{} ($0.0\%$ vs.\ $6.7\%$ with \rfinals{}) and reduces them from $13.3\%$ to $3.3\%$ for \mgptcodex{}.

Reductions in compiler errors also transfer to improving functional correctness.
\rimms{} attains the best functional correctness in 11 of the 14 model-dataset configurations.
The largest improvements are visible on \dupgrade{} for \mglm{} (from $53.3\%$ with \rfinals{} to $71.7\%$ with \rimms{}) and on \dtranslation{} for \mkimi (from $39.9\%$ to $53.9\%$, respectively).

\paragraph{Decreased Runtime}
Despite intervening during generation, \rimms{} decreases overall runtime. The average overhead over \vanilla{} drops from $233$ seconds ($+283\%$) with \rfinals{} to $135$ seconds ($+170\%$) with \rimms{}. This reduction is driven by models that require many feedback rounds, namely \mgemini{} and the \mqwen{} models. For \mqwensmall{}, \rimms{} more than halves the average runtime on both datasets; for example, on \dtranslation{}, runtime drops from on average $879$ to $357$ seconds per sample.
By interrupting generation at the first unrecoverable error, \rimms{} avoids repeatedly completing files that cannot compile. For \mopusshort{}, and \mkimi{}, which require on-the-fly feedback less frequently, runtime remains essentially unchanged, increasing by $2$ seconds ($+6\%$) for former, and decreasing $11$ seconds ($-6\%$) for the latter.

\input{figures/figure_error_analysis.tex}

\subsection{Analysis: What Effects Does Generative Compilation Have?}
\label{sec:eval:analysis}

\paragraph{Mitigated Error Snowballs}
Looking more closely at \rimms{}, we find that most errors are resolved during the early-feedback phase, i.e., within the first $k=10$ restarts.
Overall, $85.3\%$ of tasks complete without ever switching back to post-generation compiler feedback, and $55.4\%$ of all tasks are solved correctly within the early-feedback phase.
Weaker models require more restart: \mqwensmall{} takes more than 10 restarts in $40.9\%$ of tasks and \mqwenlarge{} on $20.5\%$, compared to only $2.1\%$ for \mgptcodex{}.
We attribute this partly to weaker models unfamiliarity with the early compiler-feedback format; in particular, they often regenerate the same erroneous prefix.

Generative compilation also makes error reports smaller and more focused.
When a completed file is compiled, one core mistake can cascade into many follow-on diagnostics.
By interrupting generation at the first unrecoverable mistake, \rimms{} instead surfaces a reduced set of diagnostics that focuses on the root cause.
As shown in \cref{fig:error-size}, $65\%$ of \rimms{} error reports contain only one or two distinct diagnostics, with an average of $5.5$ diagnostics per message. In contrast, only $40\%$ of \rfinals{} reports are this succinct, averaging $13.8$ diagnostics and, on \dtranslation{}, sometimes reaching several hundred.
This reduction gives the model a smaller and more targeted repair signal, reducing the burden of locating the relevant error within its context.

\paragraph{Reporting Errors Close to Their Source}
To further assess the quality of \gc{}'s error detection, we measure how early it surfaces errors during generation.
For each of the $961$ erroneous files generated by \vanilla{} and included in this analysis, we record the location of the first compiler error returned by \rustc{} in the completed program.
We then replay the same file prefix by prefix against \gc{} and record where \gc{} first reports an error.
The replay feeds each file verbatim at 50 tokens per second, a typical decoding speed for current frontier coding models\footnote{As reported by OpenRouter for \mopus{} (\url{https://openrouter.ai/anthropic/claude-opus-4.8}) and \mgptcodex{} (\url{https://openrouter.ai/openai/gpt-5.3-codex}).}, while running \gc{} in parallel without asking the model to repair errors.

\textsc{Oracle} denotes the primary span of the first compiler error reported on the completed file.
This span is a lower bound on the earliest possible detection point for that error, but is generally not achievable from a prefix alone: a prefix ending at the span may still admit a valid completion, for example if the error refers to an item that could be defined later.
We compare this oracle span with those reported by \gc{}.
We also compare against two baselines: \rimmfn{}, a simpler variant that checks only completed function bodies, and \pc{}, which reports compiler feedback only after the full file has been generated.

In the median case, \gc{} reports the error only $3$ lines after the primary span of the eventual compiler error, and in a quarter of all cases it reports the error exactly at that span.
By contrast, \rimmfn{} and \pc{} report errors much later, with median delays of $14$ and $89$ lines after \textsc{Oracle}, respectively, as shown in \cref{fig:error-delay}.
The delay distribution for \gc{} is heavy-tailed, with a mean of $27.7$ lines.
This is expected: our configuration deliberately tolerates references to not-yet-defined items, as they may legitimately be introduced later in the file.
Errors involving such items are therefore reported only once such a later definition is ruled out, e.g., when the model signals the end of file generation.
Indeed, among errors in the top quartile of detection delay ($\geq 25$ lines; mean $99$ lines), $90.5\%$ involve such items: unresolved names and imports ($63.8\%$) or missing methods and functions ($26.7\%$).

\paragraph{Errors are Detected Long before File Completion}
Reporting errors close to their source is most useful when those sources occur early in generation, since it prevents the model from producing later code based on wrong assumptions.
As shown in \cref{fig:error-cdf}, \gc{} reports unrecoverable errors at a mean detection point of $33.3\%$ of the file, nearly matching the timing-free upper bound of $32.7\%$.
Thus, \gc{} avoids generating the remaining $66.7\%$ of unrecoverably non-compiling output.
By contrast, \pc{} reports errors only after file completion by construction, while \rimmfn{} lags behind with a mean detection point of $40.3\%$.

\paragraph{Detected Error Kinds}
We categorize errors detected during mid-file restarts into syntax errors, type errors, borrow-check and lifetime errors, and others. Most errors are type-system related; the most frequent is type mismatch (E0308), which accounts for over a third of reports. Our method also detects more intricate borrow-check and lifetime violations early, including conflicting borrows (E0502) and moves out of borrowed values (E0507). A detailed breakdown appears in \cref{app:error-kinds}.

%% file: figures/table_downstream.tex
\begin{table}[t]
\centering
\caption{Comparison of \vanilla{}, \rfinals{}, and \rimms{} on two repository-level coding tasks (\dtranslation{} and \dupgrade{}). Here, \vanilla{} denotes generation without compiler feedback, \rfinals{} adds standard post-generation feedback, and \rimms{} further incorporates on-the-fly feedback enabled by generative compilation. Overall, \rimms{} further reduces compiler errors while improving functional correctness. We mark the best result in \textbf{boldface} and \underline{underline} results that are not significantly worse, determined by a paired difference test at significance level $\alpha=5\%$. }
\vspace{-1em}
\label{tab:downstream}
\setlength{\tabcolsep}{4pt}
\resizebox{\linewidth}{!}{%
\begin{tabular}{@{}lcccccc cccccc@{}}
\toprule
             & \multicolumn{6}{c}{\dtranslation{}} & \multicolumn{6}{c}{\dupgrade{}} \\
\cmidrule(lr){2-7}\cmidrule(l){8-13}
             & \multicolumn{3}{c}{Compiler Errors ($\downarrow$)} & \multicolumn{3}{c}{Functional Corr. ($\uparrow$)} & \multicolumn{3}{c}{Compiler Errors ($\downarrow$)} & \multicolumn{3}{c}{Functional Corr. ($\uparrow$)} \\
\cmidrule(lr){2-4}\cmidrule(lr){5-7}\cmidrule(lr){8-10}\cmidrule(l){11-13}
Model        & \vanillas{}   & \rfinals{}    & \rimms{}          & \vanillas{}   & \rfinals{}           & \rimms{}          & \vanillas{}   & \rfinals{}           & \rimms{}          & \vanillas{}   & \rfinals{}           & \rimms{}          \\
\midrule
\mopusshort & $51.8$\% & $7.5$\% & $\mathbf{2.6}$\% & $32.0$\% & $\underline{61.0}$\% & $\mathbf{62.3}$\% & $51.7$\% & $6.7$\% & $\mathbf{0.0}$\% & $43.3$\% & $\underline{85.0}$\% & $\mathbf{86.7}$\% \\
\mgpt & $46.1$\% & $3.9$\% & $\mathbf{1.8}$\% & $35.5$\% & $\underline{61.0}$\% & $\mathbf{61.8}$\% & $65.0$\% & $13.3$\% & $\mathbf{3.3}$\% & $35.0$\% & $75.0$\% & $\mathbf{80.0}$\% \\
\mgeminishort & $61.4$\% & $14.0$\% & $\mathbf{10.5}$\% & $23.7$\% & $\underline{50.4}$\% & $\mathbf{52.2}$\% & $68.3$\% & $\underline{16.7}$\% & $\mathbf{13.3}$\% & $26.7$\% & $\underline{78.3}$\% & $\mathbf{80.0}$\% \\
\midrule
\mkimi & $65.8$\% & $38.6$\% & $\mathbf{11.0}$\% & $23.7$\% & $39.9$\% & $\mathbf{53.9}$\% & $58.3$\% & $13.3$\% & $\mathbf{3.3}$\% & $36.7$\% & $\underline{76.7}$\% & $\mathbf{83.3}$\% \\
GLM 5.2 & $60.5$\% & $\underline{17.1}$\% & $\mathbf{15.4}$\% & $25.4$\% & $\mathbf{51.3}$\% & $\underline{50.4}$\% & $80.0$\% & $45.0$\% & $\mathbf{16.7}$\% & $20.0$\% & $53.3$\% & $\mathbf{71.7}$\% \\
\mqwenlarge & $64.5$\% & $\underline{13.2}$\% & $\mathbf{12.7}$\% & $23.2$\% & $\mathbf{54.8}$\% & $\underline{53.1}$\% & $73.3$\% & $21.7$\% & $\mathbf{15.0}$\% & $26.7$\% & $\underline{68.3}$\% & $\mathbf{71.7}$\% \\
\mqwensmall & $85.5$\% & $\underline{36.0}$\% & $\mathbf{35.1}$\% & $8.3$\% & $\underline{30.3}$\% & $\mathbf{33.3}$\% & $90.0$\% & $\mathbf{43.3}$\% & $\mathbf{43.3}$\% & $10.0$\% & $\mathbf{48.3}$\% & $\underline{41.7}$\% \\

\bottomrule
\end{tabular}%
}
\end{table}

%% file: figures/figure_error_analysis.tex
\begin{figure}[t]
\centering
\begin{subfigure}{0.32\linewidth}
\centering
\includegraphics[width=\linewidth]{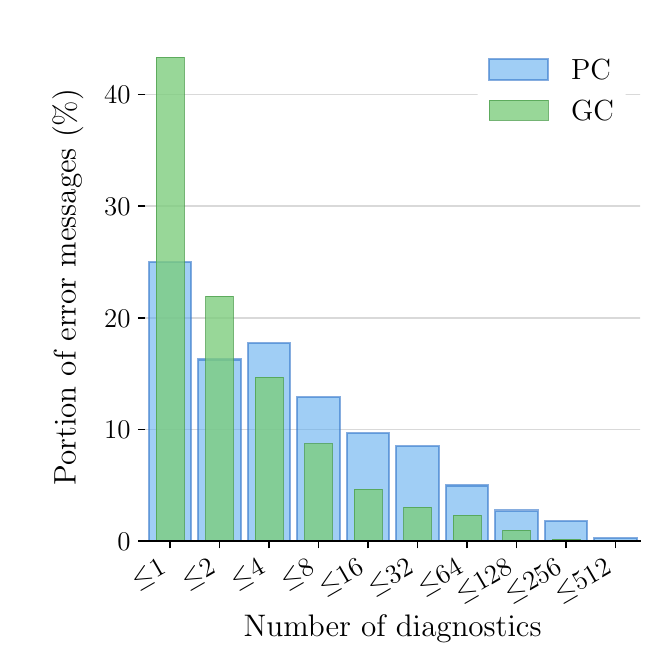}
\caption{Diagnostics per error message.}
\label{fig:error-size}
\end{subfigure}
\hfill
\begin{subfigure}{0.32\linewidth}
\centering
\includegraphics[width=\linewidth]{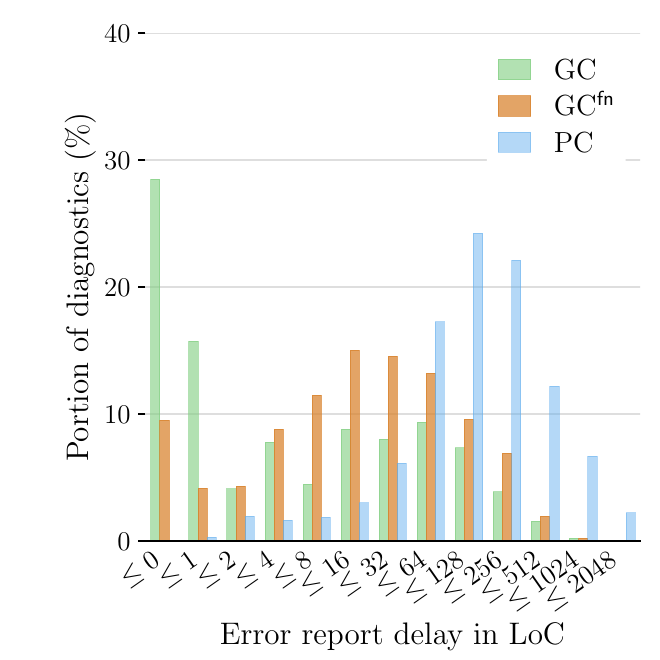}
\caption{Error report delays.}
\label{fig:error-delay}
\end{subfigure}
\hfill
\begin{subfigure}{0.32\linewidth}
\centering
\includegraphics[width=\linewidth]{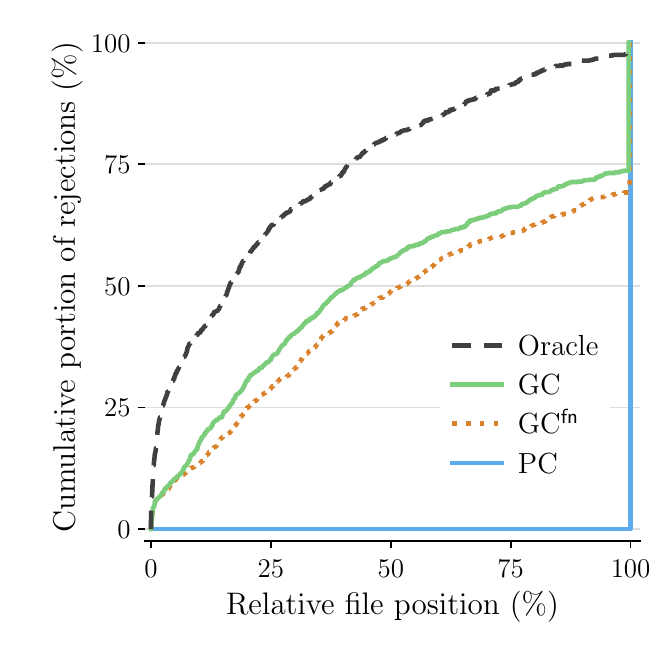}
\caption{Cumulative error detection.}
\label{fig:error-cdf}
\end{subfigure}
\vspace{-2mm}
\caption{Our analysis on the effects of generative compilation (\rimms{}). First (a), \rimms{} leads to smaller numbers of diagnostics in error messages, compared to post-generation compilation (\rfinals). Second (b), \gc{} results in smaller error report delays, compared with the function-level invariant \rimmfn{} and \pc{}. Finally (c), \gc{} detects errors early, most of the time before file completion.}
\label{fig:error}
\end{figure}

%% file: sections-new/discussion.tex
\section{Discussion and Future Work}
\label{sec:dis}

\paragraph{Modeling Error Messages}
Our generative compiler returns both a verdict $\ok$ and a diagnostic $\err$.
However, our completeness and soundness results concern only $\ok$, and do not model $\err$. In particular, we do not formalize whether $\err$ describes a genuine defect in the original, unsealed partial program rather than an artifact of sealing, but observe this behavior empirically.
This is due to two core reasons: First, we take active measures against reporting artifacts of sealing, by suppressing errors that stem from sealing-inserted code using the expected-error machinery described in \cref{sec:expected-errors}.
Second, in most cases, the placeholders that sealing introduces add no new typing obligations of their own, leaving only already-generated code that is mostly preserved verbatim.
Formalizing that $\err$ represents a genuine defect in the partial prefix, building on prior formal treatments of type-error localization~\citep{DBLP:conf/popl/Wand86,DBLP:conf/pldi/LernerFGC07,DBLP:journals/pacmpl/SeidelSCWJ17}, is future work.
A related, open question is how to design error messages for their consumer, for instance whether diagnostics should be reformulated or supplemented with context to be more actionable.
This has been studied for human programmers~\citep{DBLP:conf/sigcse/MarceauFK11} and, more recently, for AI coding agents~\citep{DBLP:journals/corr/abs-2606-01522}.
Our current rendering follows the \rustc{} diagnostic formatting, because we assume that models are most familiar with this feedback based on their training corpus.

\paragraph{When and How to Invoke Generative Compilation}
In \cref{sec:decoding}, generative compilation is tied to a fixed mechanism: invoke it whenever the previous check finishes, and on rejection, augment the prompt and restart generation.
Concurrency makes this loop practical, since validation runs alongside token generation and adds only modest overhead.
Decoupling generative compilation from this hardwired decoding loop opens several directions.
First, the invocation frequency could be tuned to trade compiler overhead against wasted generation, since less frequent checks only risk discarding a short suffix after rejection.
Second, the decision of when and how to invoke generative compilation could be delegated to the coding agent, alongside other tools it may call.
Third, our current use of generative compilation treats the LLM as a black box, leaving open whether white-box model access could yield further improvements.
Fourth, the diagnostics $\err$ produced on rejection provide valuable training signals, suggesting the possibility of using them during training rather than only at inference time.
Finally, we highlight that generative compilation is a general concept applicable beyond LLMs, and works with any code generation method that is steerable by natural language prompts, and produces intermediate code outputs.

\paragraph{Automating Sealor Construction}
The sealors in this paper, for both \FR (\cref{sec:method}) and Rust (\cref{sec:rust}), are hand-written.
For each construct, we manually derive a sealing rule and argue its completeness and soundness.
An open direction is to synthesize such rules automatically from a language's specification or reference implementation, rather than by hand.
This is a challenging task, since rules must be complete while achieving strong soundness, so synthesis would need to search under these constraints rather than under well-formedness alone.
Automating this construction could help generative compilation scale even further: to languages beyond Rust, to languages as they evolve, and even to new languages as they are designed and implemented.
The latter two cases may benefit most, since LLMs have little or no training data for such code, making direct and on-the-fly compiler feedback especially valuable.

%% file: sections-new/related.tex
\section{Related Work}

\paragraph{Compiler Feedback}
Compilers help programming in two distinct ways: they establish guarantees, such as memory safety~\citep{jim2002cyclone,matsakis2014rust} or information-flow security~\citep{myers1999jflow}, and they provide feedback that guides revision when those guarantees are violated.
Such feedback, especially well-designed error messages, has long helped human programmers from novices to professionals~\citep{DBLP:conf/popl/Wand86,DBLP:conf/pldi/LernerFGC07,DBLP:journals/pacmpl/SeidelSCWJ17,DBLP:conf/sigcse/MarceauFK11}.
Earlier research has brought this feedback into AI-assisted code generation, using it during training through supervised fine-tuning~\citep{wang-etal-2022-compilable} or reinforcement learning~\citep{dou-etal-2024-stepcoder}, and at inference time to iteratively refine generated programs~\citep{bi-etal-2024-iterative}.
Compiler feedback is now common in coding agents, which run compilers through the command line or invoke them as tools, obtain diagnostics, and accordingly revise code~\citep{DBLP:conf/icse/DeligiannisLMPR25}.
These approaches all use the compiler after a complete program has been generated.
Generative compilation instead invokes the compiler during generation, on still-partial programs, making it an earlier in-the-loop complement to post-generation feedback.

\paragraph{Constrained Decoding}
Like generative compilation, constrained decoding checks programs during autoregressive generation~\citep{synchromesh,DBLP:conf/nips/AgrawalKGLR23}.
The key difference is how the check is used.
Constrained decoding intervenes inside sampling, restricting next-token choices.
Early approaches do so locally, token by token, which can harm the final program \citep{dominos-DBLP:journals/corr/abs-2403-06988,DBLP:conf/iclr/UgareGS0M25,park2024gad}; later work makes sampling aware of global constraints, but still operates at the sampling level~\citep{loula2025syntactic,DBLP:journals/corr/abs-2506-05754,lipkin2025awrs,DBLP:journals/corr/abs-2511-22277}.
Constrained decoding for constraints beyond syntax typically requires significant reimplementation~\citep{ts,chopchop},
and a constraining system that only accepts a subset of the language (i.e., one that is incomplete) can degrade performance~\citep{biagiola2026alignmentproblemconstrainedcode}.
For this reason no constrained decoding techniques exist for complex properties such as the Rust borrow and lifetime system that we could compare generative compilation to.

\paragraph{Syntax- and Type-Guided Program Construction}
Syntax and types have long guided the program search, especially in synthesis from formal specifications~\citep{completion-DBLP:conf/pldi/GveroKKP13,DBLP:conf/fmcad/AlurBJMRSSSTU13,DBLP:conf/pldi/OseraZ15,prog-synth-refinement-types-DBLP:conf/pldi/PolikarpovaKS16,rust-synth-DBLP:journals/pacmpl/FialaI0PS23}. Our work also reasons about syntax and types, but targets LLM-based code generation.

Designed for live programming environments like IDEs, typed holes~\citep{DBLP:conf/popl/OmarVHAH17,DBLP:journals/pacmpl/OmarVCH19} are close in spirit to our placeholders: $\eps$ for \FR (\cref{sec:fr-sealor}), and $\panicterm$ and $\missingterm$ for Rust (\cref{sec:rust-distinctions}).
The purpose of typed holes is also to make incomplete programs well-typed, permitting the extraction of static type contexts~\citep{DBLP:journals/pacmpl/BlinnLKO24}.
The difference is where the hole lives.
In contrast to our placeholders, typed holes require an extension of the language: the hole in the program is made a first-class citizen of the language, and the type system is extended to reason about it.
Our placeholders instead leverage existing elements of the language and are already handled by the language's type system.

%% file: sections-new/conclusion.tex
\section{Conclusion}

We introduced generative compilation, the first approach to bring compiler feedback into the intermediate steps of LLM-based code generation.
The central idea is to seal a partial program into a complete program that a standard compiler can check, thereby turning an off-the-shelf compiler into a prefix checker with diagnostics.
This preserves the practical advantages of compiler feedback: black-box model access, use of the real compiler, and rich error messages.
At the same time, we avoid the delayed feedback of post-generation checking and the implementation burden of constrained decoding.

We formalized this idea through sealors, proving that a sealor induces a generative compiler with the same completeness and soundness guarantees, and instantiated the framework on Featherweight Rust with a syntax-guided sealor satisfying this property, fully mechanized in Lean.
We then carried the same methodology to build the first prefix checker for real Rust, where the sealor handles Rust-specific expression and statement structure, control flow, placeholders, and future-dependent compiler errors while delegating type, borrow, and lifetime reasoning to \rustc{}.

Our evaluation on Rust code generation tasks shows that generative compilation detects many errors close to their source and long before file completion, including type, borrow-check, and lifetime errors.
Integrated into code generation, this early feedback reduces non-compiling outputs and improves functional correctness compared with standard post-hoc compiler feedback.

%% file: sections-new/acks.tex
\section*{Acknowledgments}

We thank Ralf Jung for his invaluable insights into the complexity of formalizing Rust and his encouragement to pursue the project.
We also thank Marius Arvinte, Cory Cornelius, and the other members of Intel's Security and Privacy Research Lab for helpful technical discussions.
Dawn Song and Jingxuan He were supported in part by the Defense Advanced Research Projects Agency (DARPA) Translating All C To Rust (TRACTOR) program under Agreement No. HR00112590134.
Niels Mündler-Sasahara and Martin Vechev are supported by the grant SAFEAI (Certified Safe, Fair and Robust Artificial Intelligence). The work has received funding from the Swiss State Secretariat for Education, Research and Innovation (SERI), contract no.  MB22.00088. Any opinions, findings, and conclusions expressed in this material are those of the authors, and do not necessarily reflect the views of the sponsoring entities.
This research was partially funded by the Ministry of Education and Science of Bulgaria (support for INSAIT, part of the Bulgarian National Roadmap for Research Infrastructure).

%% file: sections-new/appendix.tex
\appendix

\clearpage

\input{sections-new/appendix/app_experimental_details}
\input{sections-new/appendix/app_prompts}

\clearpage

\input{sections-new/appendix/app_prompt_figures}

%% file: sections-new/appendix/app_experimental_details.tex
\section{Experimental Details}
\label{app:exp-details}

\subsection{Construction of \dupgrade{}}
\label{app:exp-upgrade}

For the \emph{\dupgrade{}} dataset, we construct a benchmark that requires LLMs to use third-party libraries with API changes after their knowledge cutoff date. The benchmark is constructed as follows.
We first fetch the 100 most downloaded Rust crates that have had minor or major version bumps in the previous 6 months at the time of writing. We then task \textsf{Codex} with \mgpt{} to filter out the packages that had changes to their public-facing API during this version bump. We then task the model to construct realistic and sufficiently complex use cases of these libraries that require the changed API. This step results in 30 single-file code projects that we manually assess to be suitable for LLM evaluation. We then remove the implementation of the function that uses this changed API, replacing it with \texttt{unimplemented!()}. The model is then tasked to regenerate the entire file, including the corrected function bodies.
The tasks require implementing an average of 5.9 function bodies (maximum 15) and solutions averaging 3,142 characters (maximum 4,952), which are verified by an average of 3.23 test cases.

\subsection{Construction of \dtranslation{}}
\label{app:exp-translation}

The dataset is based on the CRUST-Bench benchmark~\citep{khatry2025crustbench}, consisting of 100 C-to-Rust translation tasks. We run \mgpt{} and \mopus{} on these tasks and observe that they are capable of solving around 80 of these tasks zero-shot, without any compiler feedback. We focus on the remaining 20 tasks. We split each task, which concerns translating an entire library, into individual tasks for each file in the library. The final evaluation is performed by embedding the generated file into a repository of golden solutions. Since interfaces and structures are fixed before compilation, model generations for each file should be independent of the surrounding files. We list the prompts used for generation in \cref{app:prompts}.

\subsection{Experimental Setup Details}
\label{app:exp-setup}

\paragraph{Rust Implementation}
Our implementation is based on Rust 1.95.0 \citep{rustc}. We apply a minimal patch to the Rust compiler to access type inference results to reveal defined identifiers and available functions and methods as well as the respective expected number of parameters. Our changes are supplied with the implementation as a code patch.

\paragraph{Hyperparameters}
The \mqwen{} models, \mglm{} and \mkimi{} are accessed through OpenRouter. \mgemini{} is accessed directly through the Google Vertex API. \mgptcodex{} is accessed through the OpenAI responses API. \mopus{} is accessed directly through the Anthropic API. For Opus, since the API endpoint does not permit setting a temperature, we omit it in requests. For all other models we set it to $0.6$.
The maximum token output for the \dupgrade{} task is 20,000 tokens; for the \dtranslation{} task, it is 30,000 tokens.
We retry model inference errors up to 10 times with a timeout of 600 seconds. Beyond this point we consider inference for a task as failed. This occurs rarely and affects $0.8\%$ of files in \dtranslation{} and $0.2\%$ of files in \dupgrade{}. We further observe that such failures occur for the same instances, which have extremely long input sizes.
The verifier and compilation are run on a server with 2 AMD EPYC 9655 96-core processors, totaling 192 physical and 384 logical cores.

\subsection{Detailed Analysis of Detected Error Kinds}
\label{app:error-kinds}

We categorize the error codes reported during mid-file restarts in \cref{sec:eval:analysis} into syntactic errors (Syntax), type violations (Type), borrow-check and lifetime violations (Borrow \& Lifetime), and others (Other), denoting errors without an associated error code as \ttt{*}. \cref{fig:error-kinds} shows the resulting histogram per dataset.

\begin{figure}
        \centering
        \includegraphics[width=\linewidth]{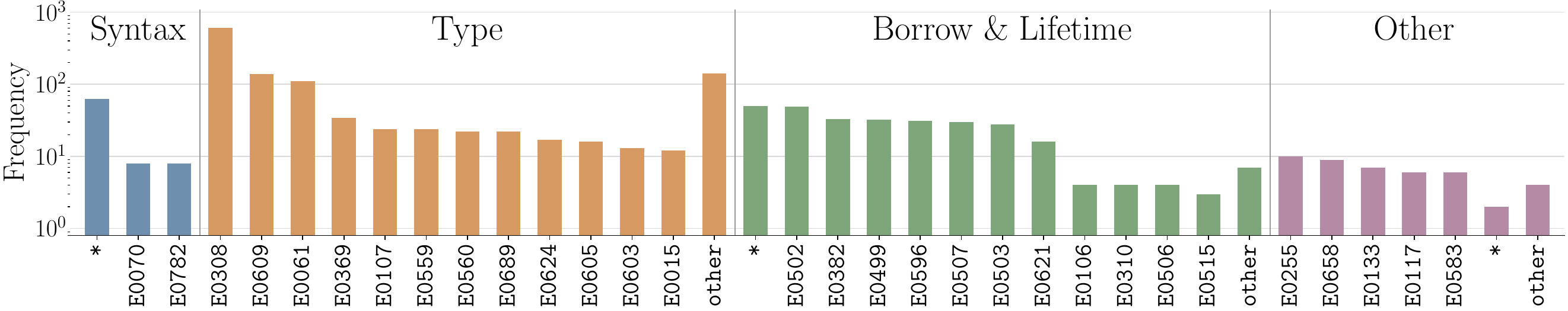}
        \caption{The error kinds detected during rollbacks span syntax errors, type errors, borrow-check and lifetime violations, and others.}
        \label{fig:error-kinds}
\end{figure}

The most frequent error on both datasets is a mismatch between an expression and its expected type (E0308), accounting for $38.2\%$ of reports on \dtranslation{} and $37.4\%$ on \dupgrade{}. On \dtranslation{}, it is followed by accesses to unknown fields ($9.5\%$, E0609), reflecting the difficulty of faithfully porting C data layouts to the provided Rust skeletons. On \dupgrade{}, the subsequent errors are calls with a wrong number of arguments ($18.1\%$, E0061), uses of a wrong number of generic arguments ($7.2\%$, E0107), initializations with nonexistent fields ($7.2\%$, E0559), and accesses to unknown fields ($5.4\%$, E0609), errors that arise when a library's public API evolves, matching the dataset design.

\subsection{Restart Budget and Token Limit}
\label{app:budget-ablation}

\begin{wrapfigure}{r}{0.45\linewidth}
        \vspace{-1em}
        \centering
        \includegraphics[width=\linewidth]{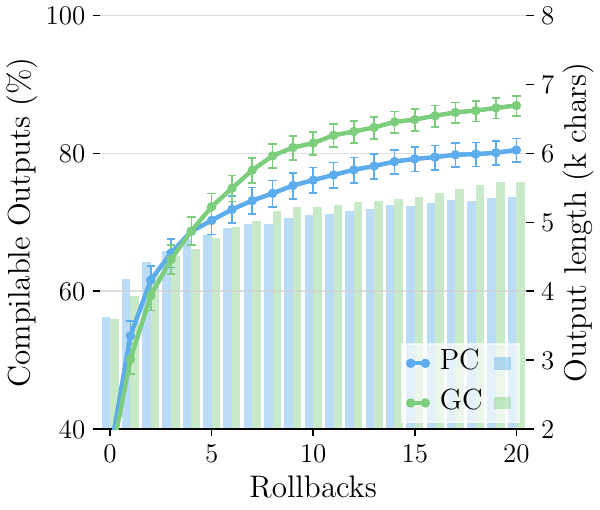}
        \caption{Fraction of compilable outputs by restart budget. Both methods flatten well before $20$ restarts. The output length (bars) stays far below the token limit.}
        \label{fig:rollback-budget}
        \vspace{-1em}
\end{wrapfigure}

To confirm that the restart budgets of $15$ (\dupgrade{}) and $20$ (\dtranslation{}) restarts and the token limits of $20,000$ and $30,000$ tokens do not constrain the results in \cref{sec:eval:downstream}, we pool the generations of all evaluated models on both datasets and measure the fraction of compilable outputs attainable within a given number of restarts (\cref{fig:rollback-budget}).
Both methods improve steeply over the first few restarts and then flatten: increasing the budget from $15$ to $20$ restarts raises the fraction of compilable outputs by at most $2.1$ percentage points for either method ($84.9\%$ to $87.0\%$ for \gc{}, $79.2\%$ to $80.5\%$ for \pc{}).
Per restart, \pc{} initially gains more, as each of its restarts has all compiler errors in context at once; from four restarts onward, however, \gc{} surpasses it. At seven restarts, \gc{} resolves $4.4$ percentage points more generations within the same budget.
Moreover, the generated outputs stay far below our token limits of $20,000$ and $30,000$ at every restart count.
We conclude that neither the restart budgets of $15$ and $20$ nor the token limits constrain the compared methods.

%
%

%% file: sections-new/appendix/app_prompts.tex
\section{End-to-End Example}
\label{app:end-to-end}

The excerpt in \cref{fig:prompt-e2e-immediate} is from the \dupgrade{} task for
\texttt{rustls-webpki}, obtained with \mopus{}.
The model is asked to update a small certificate-validation crate to a new
version of the library while filling the missing implementation in
\texttt{src/lib.rs}.
During generation, the model constructs a certificate wrapper as a temporary
and stores a trust anchor that borrows from it.
Our setup detects that the temporary would be dropped while the borrow remains
live and returns the \texttt{E0716} diagnostic immediately.
After receiving this feedback, the model binds the certificate wrapper to a
local variable, extending its lifetime; and is otherwise correct as well. The early feedback was provided in line 4 of a 32 line function, thus preventing potentially wasteful generation of a significant number of following code.

\section{Inference Prompts}
\label{app:prompts}

The inference code sends one user message containing the dataset prompt. The required prefix for grading is either prefilled as the beginning of the assistant answer (open-weight models) or appended to the user message as an instruction to start with that exact text (black-box models).

\subsection{Generic Inference Wrappers}
\label{app:prompts-generic}

Since we require a specific prefix for grading of results but cannot prefill assistant answers for black-box models, we append the short instruction shown in \Cref{fig:prompt-empty-initial-answer} that specifies the required answer prefix; it appears instantiated at the end of the dataset prompts in \Cref{fig:prompt-apiupgrade-user,fig:prompt-crustbench-user}.
\Cref{fig:prompt-immediate-rollback} is the user message appended by \rimm{} after a rejected partial output; the rejected output itself is added to the prompt history as an assistant message, and generation restarts from the prefilled prefix.
\Cref{fig:prompt-final-edit-feedback} is the corresponding user message used by \rfinal{} after a rejected complete output. Here, the model responds with complete replacement files, patches in a unified diff-like format, or a combination of both, which are applied to the previous output and re-checked. If a returned patch fails to apply, a follow-up user message reports the patch application error and repeats the format instructions.

\subsection{\dupgrade{}}
\label{app:prompts-apiupgrade}

\Cref{fig:prompt-apiupgrade-user} is the generation prompt used for
\dupgrade{} inference. It is sent as the user message for each crate
instance, with the concrete file list and incomplete \texttt{src/lib.rs}
substituted into the placeholders.

\subsection{\dtranslation{}}
\label{app:prompts-crustbench}

\Cref{fig:prompt-crustbench-user} is the user prompt used for
\dtranslation{} inference. It is instantiated with the C sources and the
incomplete Rust interface of the target file.

%% file: sections-new/appendix/app_prompt_figures.tex
\begin{figure}[p]
\centering
\begin{tcolorbox}[colback=white,colframe=black,boxrule=0.4pt,arc=0pt,left=1mm,right=1mm,top=1mm,bottom=1mm]
\noindent\hfill\textsf{\scriptsize Prompt}\par\vspace{-1em}
\begin{lstlisting}[language={},basicstyle=\ttfamily\tiny,breaklines=true,numbers=none,mathescape=false,escapechar=~]
You are given a Rust crate.
Implement only `src/lib.rs` based on the crate's documentation comments, inline comments, and tests.
[...]
/// Validates a certificate chain for server use.
pub fn validate_server_usage(
    end_entity_der: &[u8], trust_anchor_der: &[u8], intermediate_ders: &[Vec<u8>]
) -> Result<CheckReport, CheckError> {
    unimplemented!();
}
\end{lstlisting}
\par\nointerlineskip\vspace{1mm}\noindent\hspace*{-1mm}\rule{\dimexpr\linewidth+2mm\relax}{0.4pt}\par\vspace{1mm}

\noindent\hfill\textsf{\scriptsize Assistant}\par\vspace{-1em}
\begin{lstlisting}[language={},basicstyle=\ttfamily\tiny,breaklines=true,numbers=none,mathescape=false,escapechar=~]
[...]
pub fn validate_server_usage(
    end_entity_der: &[u8], trust_anchor_der: &[u8], intermediate_ders: &[Vec<u8>]
) -> Result<CheckReport, CheckError> {
    let intermediates_seen = intermediate_ders.len();

    let ee_der = CertificateDer::from(end_entity_der);
    let end_entity = EndEntityCert::try_from(&ee_der).map_err(CheckError::Certificate)?;

    let trust_anchor = webpki::anchor_from_trusted_cert(&CertificateDer::from(trust_anchor_der))
        .map_err(CheckError::Certificate)?;
    let trust_anchors: Vec<TrustAnchor> = vec![trust_anchor];

    let intermediate_certs: Vec
\end{lstlisting}
\par\nointerlineskip\vspace{1mm}\noindent\hspace*{-1mm}\rule{\dimexpr\linewidth+2mm\relax}{0.4pt}\par\vspace{1mm}

\noindent\hfill\textsf{\scriptsize Compiler Feedback}\par\vspace{-1em}
\begin{lstlisting}[language={},basicstyle=\ttfamily\tiny,breaklines=true,numbers=none,mathescape=false,escapechar=~]
Your previous partial output was rejected by the compiler.
Compiler error details:
error[E0716]: temporary value dropped while borrowed
 --> src/lib.rs:155:43
    |
155 |         webpki::anchor_from_trusted_cert(&CertificateDer::from(trust_anchor_der))
    |                                           ^^^^^^^^^^^^^^^^^^^^^^^^^^^^^^^^^^^^^^ creates a temporary value which
    |                                                                                     is freed while still in use
156 |         .map_err(CheckError::Certificate)?;
    |                                          - temporary value is freed at the end of this statement
157 |     let trust_anchors: Vec<TrustAnchor> = vec![trust_anchor];
    |                                                ------------ borrow later used here
  = help: consider using a `let` binding to create a longer lived value

Restart the generation.
\end{lstlisting}
\par\nointerlineskip\vspace{1mm}\noindent\hspace*{-1mm}\rule{\dimexpr\linewidth+2mm\relax}{0.4pt}\par\vspace{1mm}

\noindent\hfill\textsf{\scriptsize Assistant}\par\vspace{-1em}
\begin{lstlisting}[language={},basicstyle=\ttfamily\tiny,breaklines=true,numbers=none,mathescape=false,escapechar=~]
[...]
pub fn validate_server_usage(
    end_entity_der: &[u8], trust_anchor_der: &[u8], intermediate_ders: &[Vec<u8>]
) -> Result<CheckReport, CheckError> {
    let intermediates_seen = intermediate_ders.len();

    let ee_der = CertificateDer::from(end_entity_der);
    let end_entity = EndEntityCert::try_from(&ee_der).map_err(CheckError::Certificate)?;

    let trust_anchor_cert = CertificateDer::from(trust_anchor_der);
    let trust_anchor = webpki::anchor_from_trusted_cert(
        &trust_anchor_cert
    ).map_err(CheckError::Certificate)?;
    let trust_anchors: Vec<TrustAnchor> = vec![trust_anchor];

    let intermediate_certs: Vec<CertificateDer> = intermediate_ders
        .iter()
        .map(|der| CertificateDer::from(der.as_slice()))
        .collect();

    let now = unix_time_now()?;
    let policy = UsagePolicy::server();
    [...]
\end{lstlisting}

%
%
%
\end{tcolorbox}
\caption{End-to-end excerpt from a \gc{} run by \mopus{} on the \texttt{rustls-webpki} API-update task. The partial output borrows from a temporary \texttt{CertificateDer} and is rejected with \texttt{E0716}. After receiving the diagnostic, the model introduces the longer-lived \texttt{trust\_anchor\_cert} binding and generates semantically valid code.}
\label{fig:prompt-e2e-immediate}
\end{figure}

\begin{figure}[t]
\centering
\begin{tcolorbox}[colback=white,colframe=black,boxrule=0.4pt,arc=0pt,left=1mm,right=1mm,top=1mm,bottom=1mm]
\begin{lstlisting}[language={},basicstyle=\ttfamily\scriptsize,breaklines=true,numbers=none,mathescape=false,escapechar=~]
Start your response with this exact prefix:
{prefilled_answer}
\end{lstlisting}
\end{tcolorbox}
\caption{Instruction appended to the user message when inference is run with an empty initial assistant answer. It replaces the usual prefilled assistant prefix by requiring the model to emit that prefix itself.}
\label{fig:prompt-empty-initial-answer}
\end{figure}

\begin{figure}[t]
\centering
\begin{tcolorbox}[colback=white,colframe=black,boxrule=0.4pt,arc=0pt,left=1mm,right=1mm,top=1mm,bottom=1mm]
\begin{lstlisting}[language={},basicstyle=\ttfamily\scriptsize,breaklines=true,numbers=none,mathescape=false,escapechar=~]
Your previous partial output was rejected by the compiler.
Compiler error details:
{rendered_compiler_details}

Restart the generation.
\end{lstlisting}
\end{tcolorbox}
\caption{Feedback message used by \rimm{}. The prompt history first receives the rejected partial assistant output \texttt{\{prefix\}\{code\}}, followed by this user message with rendered compiler diagnostics.}
\label{fig:prompt-immediate-rollback}
\end{figure}

\begin{figure}[t]
\centering
\begin{tcolorbox}[colback=white,colframe=black,boxrule=0.4pt,arc=0pt,left=1mm,right=1mm,top=1mm,bottom=1mm]
\begin{lstlisting}[language={},basicstyle=\ttfamily\scriptsize,breaklines=true,numbers=none,mathescape=false,escapechar=~]
Your previous output was rejected by the compiler.
Compiler error details:
{rendered_compiler_details}

Return only the changes needed to fix the current project output. You may use any combination of these two formats in the same answer:

1. Regenerate complete Rust files that should replace files from the current project state. Use this format for each regenerated file:
{path/to/file.rs}
```rust
// complete replacement Rust file
```

2. Return one or more patches that edit the previous output. Use this exact patch format:
*** Begin Patch
*** Update File: path/to/file.rs
@@
-old line
+new line
*** End Patch

Omitted files stay fixed exactly as currently generated.
\end{lstlisting}
\end{tcolorbox}
\caption{Feedback message used by \rfinal{}. The prompt history first receives the rejected complete assistant output, followed by this user message with the rendered compiler diagnostics. The model may respond with full replacement files, patches, or a combination; the result is applied and re-checked.}
\label{fig:prompt-final-edit-feedback}
\end{figure}

\begin{figure}[t]
\centering
\begin{tcolorbox}[colback=white,colframe=black,boxrule=0.4pt,arc=0pt,left=1mm,right=1mm,top=1mm,bottom=1mm]
\begin{lstlisting}[language={},basicstyle=\ttfamily\scriptsize,breaklines=true,numbers=none,mathescape=false,escapechar=~]
You are given a Rust crate.
Implement only `{interface_path}` based on the crate's documentation comments, inline comments, and tests.

Please follow these instructions:
  - Implement all `unimplemented!()` sections in `{interface_path}`.
  - Keep function signatures unchanged.
  - Produce compiling Rust code.
  - Do not use `unsafe`.
  - Return only the completed `{interface_path}` file contents.

Return only:
  {interface_path}
  ```rust
  // completed Rust code
  ```

These are the files:

{path_1}
```{language_1}
{code_1}
```

...

{path_n}
```{language_n}
{code_n}
```

{interface_path} (incomplete)
```rust
{interface_code}
```

Start your response with this exact prefix:
{interface_path}
```rust
\end{lstlisting}
\end{tcolorbox}
\caption{User prompt for \dupgrade{}. The model receives the crate files, including tests and comments, and is asked to implement only the incomplete interface file; \texttt{\{interface\_path\}} is instantiated with \texttt{src/lib.rs}. The trailing prefix instruction replaces the prefilled assistant answer for black-box models (\cref{app:prompts-generic}).}
\label{fig:prompt-apiupgrade-user}
\end{figure}

\begin{figure}[t]
\centering
\begin{tcolorbox}[colback=white,colframe=black,boxrule=0.4pt,arc=0pt,left=1mm,right=1mm,top=1mm,bottom=1mm]
\begin{lstlisting}[language={},basicstyle=\ttfamily\scriptsize,breaklines=true,numbers=none,mathescape=false,escapechar=~]
You are an expert at converting C To Rust.
You will be provided with C source files in the format:

{filename.c / filename.h}
```c
// Input C code
```

You will also be given the Rust Interface definition corresponding to the C file that you must implement by filling in the unimplemented!() parts based on the C code.
The Rust Interface file will be in the format:

{filename.rs}
```rust
// Rust Interface code
```

I need you to transpile the provided code files from C to Rust.
Please follow the instructions provided below:
  - Each C file I provide MUST be transpiled into a corresponding Rust file.
  - You MUST use the same name as the Rust interface file name, when generating the transpiled Rust Code
  - You MUST always implement the Rust Interface files using the C code as reference.
  - You MUST ensure that you implement ALL unimplemented!() parts strictly.
  - You MUST NOT change the function signatures of the Rust code.
  - Each transpiled Rust file MUST compile.
  - The transpiled Rust code MUST be observationally equivalent to the C code.
  - The transpiled Rust code MUST NOT contain Foreign Function Interface calls, such as the libc library.
  - The transpiled Rust code MUST NOT contain unsafe blocks.
  - You must ensure that you import the required files that are referenced in each Rust file.
  - You may derive traits for the data types that are defined, to elicit required behavior during implementation. For instance you can add #[derive(Clone)] for structs when the Clone trait is required for implementation.
  - All imports in the rust file must be in the following format -
    ```rust
      use crate::file_name::module;
    ```

Please think step-by-step and return your final solution for the transpiled file in the following format:

  {filename.rs}
  ```rust
  // Generated Rust Code
  ```

        These are the files:
{c_sources}

{interface_path} (incomplete)
```rust
{interface_code}
```

Start your response with this exact prefix:
{interface_path}
```rust
\end{lstlisting}
\end{tcolorbox}
\caption{User prompt for \dtranslation{}. It gives the C sources and the incomplete Rust interface of the target file. The trailing prefix instruction replaces the prefilled assistant answer for black-box models (\cref{app:prompts-generic}).}
\label{fig:prompt-crustbench-user}
\end{figure}